\newif\ifarxiv\arxivtrue%

\documentclass[envcountsame,runningheads]{llncs}

\usepackage{silence}
\usepackage{amsfonts}
\usepackage{amssymb}
\usepackage{amsmath}
\usepackage{mathtools}
\usepackage{mathpartir}
\usepackage{stmaryrd}
\usepackage{amsfonts}
\usepackage{marvosym}
\usepackage{scalerel}
\usepackage{clrscode3e}
\usepackage{todonotes}
\usepackage{multirow,bigdelim}
\usepackage{array}
\usepackage{enumitem}

\makeatletter
\RequirePackage[bookmarks,unicode,colorlinks=true]{hyperref}%
   \def\@citecolor{blue}%
   \def\@urlcolor{blue}%
   \def\@linkcolor{blue}%

\def\orcidID#1{\smash{\href{http://orcid.org/#1}{\protect\raisebox{-1.25pt}{\protect\includegraphics{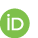}}}}}
\makeatother

\usepackage{cleveref}

\usepackage{subcaption}
\usepackage{tikz}
\usetikzlibrary{positioning}
\usetikzlibrary{automata}
\usetikzlibrary{angles}
\usetikzlibrary{calc}
\usetikzlibrary{decorations.pathmorphing}

\usepackage{newfloat}
\DeclareFloatingEnvironment[name=Algorithm,within=none]{algorithm}

\usepackage{etoolbox}
\apptocmd{\thebibliography}{\hbadness 10000\relax}{}{} % chktex 1

\newcolumntype{C}[1]{>{\centering\arraybackslash$}m{#1}<{$}}
\newcolumntype{L}[1]{>{\raggedleft\arraybackslash$}m{#1}<{$}}
\newcolumntype{R}[1]{>{\raggedright\arraybackslash$}m{#1}<{$}}

\newcommand{\runrel}{\mathrel{\rightarrow}}
\newcommand{\arun}[2][]{\mathrel{\raisebox{-3pt}{$\xrightarrow[#1]{#2}$}}}

\newcommand*{\lmbrace}{\{\mskip-4mu|}
\newcommand*{\rmbrace}{|\mskip-4mu\}}
\newcommand*{\mset}[1]{\lmbrace#1\rmbrace}
\newcommand{\angl}[1]{\left\langle#1\right\rangle}

\newcommand{\bigOh}[1]{\ensuremath{\mathcal{O}(#1)}}

\newcommand{\naturals}{\mathbb{N}}

\newcommand{\ltr}[1]{\mathtt{#1}}
\newcommand{\M}{\mathbb{M}}

\newcommand{\unit}{\mathbf{1}}
\newcommand{\zero}{\mathbf{0}}
\newcommand{\pempty}{1}
\newcommand{\cempty}{\square}
\newcommand{\pr}{\mathcal{R}}

\newcommand{\asymbol}{\mathbin{\scalerel*{\heartsuit}{\cdot}}}

\newcommand{\SP}{\ensuremath{\mathsf{SP}}}
\newcommand{\PC}{\ensuremath{\mathsf{PC}}}
\newcommand{\LStar}{\ensuremath{\mathtt{L}^{\!\star}}}

\newcommand{\lang}{\mathcal{L}}

\newcommand{\free}[1]{{#1}^\sharp}

\newcommand{\ext}[1]{{#1}^+}

\newcommand{\cinsert}[2]{#1[#2]}

\newcommand{\row}{\mathsf{row}}
\newcommand{\hyp}{\mathcal{H}}

\newcommand{\lp}[1]{\mathbf{#1}}

\makeatletter
\newenvironment{oneshot}[1]{\@begintheorem{#1.}{\unskip}}{\@endtheorem}
\makeatother

\tikzset{every state/.style={minimum size=1pt}}

\begin{document}
\title{Learning Pomset Automata\thanks{This work was partially supported by the ERC Starting Grant ProFoundNet (679127) and the EPSRC Standard Grant CLeVer (EP/S028641/1). The authors thank Matteo Sammartino for useful discussions.}}

%\titlerunning{Abbreviated paper title}

\author{%
    Gerco van Heerdt\inst{1}~\orcidID{0000-0003-0669-6865} (\Letter) % chktex 8
    \and
    Tobias Kapp\'e\inst{2}~\orcidID{0000-0002-6068-880X} % chktex 8
    \and\\
    Jurriaan Rot\inst{3}
    \and
    Alexandra Silva\inst{1}~\orcidID{0000-0001-5014-9784} % chktex 8
}
\authorrunning{Van Heerdt et al.}
\institute{%
	University College London, London, UK\\
	\email{gerco.heerdt@ucl.ac.uk} \and
	Cornell University, Ithaca NY, USA \and
	Radboud University, Nijmegen, The Netherlands%
}
\maketitle
\begin{abstract}
	We extend the {\LStar} algorithm to learn bimonoids recognising pomset languages.
	We then identify a class of pomset automata that accepts precisely the class of pomset languages recognised by bimonoids and show how to convert between bimonoids and automata.

%	\keywords{Automata learning \and pomsets \and pomset automata.}
\end{abstract}

\section{Introduction}

Automata learning algorithms are useful in automated inference of models, which is needed for verification of hardware and software systems. In \emph{active}
learning, the algorithm interacts with a system through tests and observations to produce a model of the system's behaviour. One of the first active learning algorithms proposed was \LStar, due to Dana Angluin~\cite{Angluin87}, which infers a minimal deterministic automaton for a target regular language. \LStar\ has been used in a range of verification tasks, including learning error traces in a program~\cite{DBLP:conf/atva/ChapmanCKKST15}. For more advanced verification tasks, richer automata types are needed and \LStar\ has been extended to e.g.\ input-output~\cite{fides}, register~\cite{learnlib}, and weighted automata~\cite{wfa-pid}. None of the existing extensions can be used in analysis of concurrent programs.

Partially ordered multisets (pomsets)~\cite{grabowski-1981,gischer-1988} are basic structures used in the modeling and semantics of concurrent programs. Pomsets generalise words, allowing to capture both the sequential and the parallel structure of a trace in a concurrent program. Automata accepting pomset languages are therefore useful to study the operational semantics of concurrent programs---see, for instance, work on concurrent Kleene algebra~\cite{hoare-moeller-struth-wehrman-2009,laurence-struth-2014,DBLP:conf/concur/KappeBL0Z17,DBLP:conf/esop/KappeB0Z18}.

In this paper, we propose an active learning algorithm for a class of pomset automata. The approach is algebraic: we consider languages of pomsets recognised by bimonoids~\cite{lodaya-weil-2000} (which we shall refer to as pomset recognisers). This can be thought of as a generalisation of the classical approach to language theory of using monoids as word acceptors: bimonoids have an extra operation that models parallel composition in addition to sequential. The two operations give rise to a complex branching structure that makes the learning process non-trivial.

The key observation is that pomset recognisers are tree automata whose algebraic structure satisfies additional equations. We extend tree automata learning algorithms~\cite{DrewesH03,drewes2007,Sakakibara90} to pomset recognisers. The main challenge is to ensure that intermediate hypotheses in the algorithm are valid pomset recognisers, which is essential in practical scenarios where the learning process might not run to the very end, returning an approximation of the system under learning. This requires equations of bimonoids to be correctly propagated and preserved in the core data structure of the algorithm---the observation table. The proof of termination, in analogy to \LStar, relies on the existence of a canonical pomset recogniser of a language, which is based on its syntactic bimonoid. The steps of the algorithm provide hypotheses that get closer in size to the canonical recogniser.

Finally, we bridge the learning algorithm to pomset automata~\cite{DBLP:conf/concur/KappeBL0Z17,kappe-brunet-luttik-silva-zanasi-2018b} by providing two constructions that enable us to seamlessly move between pomset recognisers and pomset automata. Note that although bimonoids provide a useful formalism to denote pomset languages, which is amenable to the design of the learning algorithm, they enforce a redundancy that is not present in pomset automata: whereas a pomset automaton processes a pomset from left to right in sequence, one letter per branch at a time, a bimonoid needs to be able to take the pomset represented as a binary tree in any way and process it bottom-up. This requirement of different decompositions leading to the same result makes bimonoids in general much larger than pomset automata and hence the latter are, in general, a more efficient representation of a pomset language.

The rest of the paper is organised as follows. We conclude this introductory section with a review of relevant related work. \Cref{sec:prelims} contains the basic definitions on pomsets and pomset recognisers. The learning algorithm for pomset recognisers appears in \Cref{sec:learning}, including proofs to ensure termination and invariant preservation. \Cref{sec:conversion} presents constructions to translate between (a class of) pomset automata and pomset recognisers. We conclude with discussion of further work in \Cref{sec:discussion}. Omitted proofs appear in
\ifarxiv%
\Cref{appendix:proofs-algo,appendix:proofs-pas}%
\else%
the extended version~\cite{arxiv-version}%
\fi.

\paragraph{Related Work.}
There is a rich literature on adaptations and extensions of \LStar\ from deterministic automata
to various kinds of models, see, e.g.,~\cite{Vaandrager17,HowarS16} for an overview.
To the best of our knowledge, this paper is the first to provide an active learning algorithm
for pomset languages recognised by finite bimonoids.

Our algorithm learns an algebraic recogniser.
Urbat and Schr{\"{o}}der~\cite{urbat2020} provide a very general learning
approach for languages recognised by algebras for monads~\cite{Bojanczyk15,UrbatACM17}, based on a reduction to categorical automata, for which they present an \LStar-type
 algorithm. Their reduction gives rise to an infinite alphabet
in general, so tailored work is needed for deriving algorithms and finite representations.
This can be done for instance for monoids, recognising regular languages,
but it is not clear how this could extend to pomset recognisers. We present a
direct learning algorithm for bimonoids, which does not rely on any encoding.

Our concrete learning algorithm for bimonoids is closely related to
learning approaches for bottom-up tree automata~\cite{DrewesH03,drewes2007,Sakakibara90}: pomset languages can be viewed as
tree languages satisfying certain equations. Incorporating these equations turned out to be a non-trivial
task, which requires additional checks on the observation table during execution of the algorithm.

Conversion between recognisers and automata for a pomset language was first explored by Lodaya and Weil~\cite{lodaya-weil-2000,lodaya-weil-1998}.
Their results relate the expressive power of these formalisms to \emph{sr-expressions}.
As a result, converting between recognisers and automata using their construction uses an sr-expression as an intermediate representation, increasing the resulting state space.
Our construction, however, converts recognisers directly to pomset automata, which keeps the state space relatively small.
Moreover, Lodaya and Weil work focus on pomset languages of \emph{bounded width}, i.e., with an upper bound on the number of parallel events.
In contrast, our conversions work for all recognisable pomset languages (and a suitable class of pomset automata), including those of unbounded width.

\'{E}sik and N\'{e}meth~\cite{esik-nemeth-2004} considered automata and recognisers for \emph{biposets}, i.e., sp-pomsets without commutativity of parallel composition.
They equate languages recognised by \emph{bisemigroups} (bimonoids without commutativity or units) with those accepted by \emph{parenthesizing automata}.
Our equivalence is similar in structure, but relates a subclass of pomset automata to bimonoids instead.
The results in this paper can easily be adapted to learn representations of biposet languages using bisemigroups, and convert those to parenthesizing automata.

\section{Pomset Recognisers}\label{sec:prelims}

Throughout this paper we fix a finite \emph{alphabet} $\Sigma$ and assume $\cempty \not\in \Sigma$.
When defining sets parameterised by a set $X$, say $\mathsf{S}(X)$, we may use $\mathsf{S}$ to refer to $\mathsf{S}(\Sigma)$.%

We recall pomsets~\cite{gischer-1988,grabowski-1981}, a generalisation of words that model concurrent traces.
A \emph{labelled poset} over $X$ is a tuple $\lp{u} = \angl{S_\lp{u}, \leq_\lp{u}, \lambda_\lp{u}}$, where $S_\lp{u}$ is a finite set (the \emph{carrier} of $\lp{u}$), $\leq_\lp{u}$
is a partial order on $S_\lp{u}$ (the \emph{order} of $\lp{u}$), and $\lambda_\lp{u} \colon S_\lp{u} \to X$ is a function (the \emph{labelling} of $\lp{u}$).
Pomsets are labelled posets up to isomorphism.

\begin{definition}[Pomsets]%
\label{definition:lp-isomorphism-pomset}
Let $\lp{u}, \lp{v}$ be labelled posets over $X$.
An \emph{embedding} of $\lp{u}$ in $\lp{v}$ is an injection $h: S_\lp{u} \to S_\lp{v}$ such that $\lambda_\lp{v} \circ h = \lambda_\lp{u}$ and $s \leq_\lp{u} s'$ if and only if $h(s) \leq_\lp{v} h(s')$.
An \emph{isomorphism} is a bijective embedding whose inverse is also an embedding.
We say $\lp{u}$ is \emph{isomorphic} to $\lp{v}$, denoted $\lp{u} \cong \lp{v}$, if there exists an isomorphism between $\lp{u}$ and $\lp{v}$.
A \emph{pomset} over $X$ is an isomorphism class of labelled posets over $X$, i.e., $[\lp{v}] = \{ \lp{u} : \lp{u} \cong \lp{v} \}$.
When $u = [\lp{u}]$ and $v = [\lp{v}]$ are pomsets, $u$ is a \emph{subpomset} of $v$ when there exists an embedding of $\lp{u}$ in $\lp{v}$.
\end{definition}

When two pomsets are in scope, we tacitly assume that they are represented by labelled posets with disjoint carriers.
We write $\pempty$ for the empty pomset.
When $\ltr{a} \in X$, we write $\ltr{a}$ for the pomset represented by the labelled poset whose sole element is labelled by $\ltr{a}$.
Pomsets can be composed in sequence and in parallel:

\begin{definition}[Pomset composition]%
\label{definition:pomset-composition}
Let $u = [\lp{u}]$ and $v = [\lp{v}]$ be pomsets over $X$.
We write $u \parallel v$ for the \emph{parallel composition} of $u$ and $v$, which is the pomset over $X$ represented by the labelled poset
\begin{align*}
    \lp{u} \parallel \lp{v} &= \angl{S_\lp{u} \cup S_\lp{v}, \;\; {\leq_\lp{u}} \cup {\leq_\lp{v}}, \;\; \lambda_\lp{u} \cup \lambda_\lp{v}}
\intertext{%
Similarly, we write $u \cdot v$ for the \emph{sequential composition} of $u$ and $v$, that is, the pomset represented by the labelled poset
}
    \lp{u} \cdot \lp{v} &= \angl{S_\lp{u} \cup S_\lp{v}, \;\; {\leq_\lp{u}} \cup {\leq_\lp{v}} \cup S_\lp{u} \times S_\lp{v}, \;\; \lambda_\lp{u} \cup \lambda_\lp{v}}
\end{align*}
We may elide the dot for sequential composition, for instance writing $\ltr{a}\ltr{b}$ for $\ltr{a} \cdot \ltr{b}$.
\end{definition}

The pomsets we use can be built using sequential and parallel composition.

\begin{definition}[Series-parallel pomsets]%
\label{definition:pomset-sp}
The set of \emph{series-parallel pomsets (sp-pomsets)} over $X$, denoted $\SP(X)$, is the smallest set such that $\pempty \in \SP(X)$ and $\ltr{a} \in \SP(X)$ for every $\ltr{a} \in X$, closed under parallel and sequential composition.
\end{definition}

Concurrent systems admit executions of operations that are not only ordered in sequence but also allow parallel branches.
An algebraic structure consisting of both a sequential and a parallel composition operation, with a shared unit, is called a \emph{bimonoid}.
Formally, its definition is as follows.

\begin{definition}[Bimonoid]
	A \emph{bimonoid} is a tuple $\angl{M, \odot, \obar, \unit}$ where
	\begin{itemize}
		\item
			$M$ is a set called the \emph{carrier} of the bimonoid,
		\item
			$\odot$ is a binary associative operation on $M$,
		\item
			$\obar$ is a binary associative and commutative operation on $M$, and
		\item
			$\unit \in M$ is a unit for both $\odot$ (on both sides) and $\obar$.
	\end{itemize}
	Bimonoid homomorphisms are defined in the usual way.
\end{definition}

Given a set $X$, the \emph{free bimonoid}~\cite{gischer-1988} over $X$ is $\angl{\SP(X), \cdot, \parallel, \pempty}$.
The fact that it is free means that for every function $f \colon X \to M$ for a given bimonoid $\angl{M, \odot, \obar, \unit_M}$ there exists a unique bimonoid homomorphism $\free{f} \colon \SP(X) \to M$ such that the restriction of $\free{f}$ to $X$ is $f$.

Just as monoids can recognise words, bimonoids can recognise pomsets~\cite{lodaya-weil-2000}.
A bimonoid together with the witnesses of recognition is a \emph{pomset recogniser}.

\begin{definition}[Pomset recogniser]
	A \emph{pomset recogniser} is a tuple $\pr = \angl{M, \odot, \obar, \unit, i, F}$ where $\angl{M, \odot, \obar, \unit}$ is a bimonoid, $i \colon \Sigma \to M$, and $F \subseteq M$.
	The \emph{language recognised} by $\pr$ is given by
	\(
		\lang_\pr = \{u \in \SP : \free{i}(u) \in F\} \subseteq \SP.
	\)
\end{definition}

\begin{example}%
\label{ex:loop}
Suppose a program consists of a loop, where each iteration runs actions $\ltr{a}$ and $\ltr{b}$ in parallel.
We can describe the behaviour of this program by
\[
    \lang
        = {\{ \ltr{a} \parallel \ltr{b} \}}^*
        = \{ \pempty, \ltr{a} \parallel \ltr{b}, (\ltr{a} \parallel \ltr{b}) \cdot (\ltr{a} \parallel \ltr{b}), \ldots \}
\]

We can describe this language using a pomset recogniser, as follows.
Let $M = \{ q_\ltr{a}, q_\ltr{b}, q_1, q_\bot, \unit \}$, and let $\odot$ and $\obar$ be the operations on $M$ given by
\begin{mathpar}
q \odot q' =
    \begin{cases}
    q & q' = \unit \\
    q' & q = \unit \\
    q_1 & q = q' = q_1 \\
    q_\bot & \mathrm{otherwise}
    \end{cases}
\and
q \obar q' =
    \begin{cases}
    q & q' = \unit \\
    q' & q = \unit \\
    q_1 & \{ q, q' \} = \{ q_\ltr{a}, q_\ltr{b} \} \\
    q_\bot & \mathrm{otherwise}
    \end{cases}
\end{mathpar}
A straightforward proof verifies that $\angl{M, \odot, \obar, \unit}$ is a bimonoid.

We set $i(\ltr{a}) = q_\ltr{a}$, $i(\ltr{b}) = q_\ltr{b}$, and $F = \{ \unit, q_1 \}$.
Now, for $n > 0$:
\begin{align*}
    i^\sharp(\underbrace{(\ltr{a} \parallel \ltr{b}) \cdots (\ltr{a} \parallel \ltr{b})}_{\text{$n$ times}})
        = \underbrace{(i(\ltr{a}) \parallel i(\ltr{b})) \odot \cdots \odot (i(\ltr{a}) \parallel i(\ltr{b}))}_{\text{$n$ times}}
        = \underbrace{q_1 \odot \cdots \odot q_1}_{\text{$n$ times}}
        = q_1
\end{align*}
No other pomsets are mapped to $q_1$; hence, $\angl{M, \odot, \obar, \mathbf{1}, i, F}$ accepts $\mathcal{L}$.
\end{example}

\begin{example}%
\label{ex:nested}
Suppose a program solves a problem recursively, such that the recursive calls are performed in parallel.
In that case, the program would either perform the base action $\ltr{b}$, or some preprocessing action $\ltr{a}$ followed by running two copies of itself in parallel.
This behaviour can be described by the smallest pomset language $\lang$ satisfying the following inference rules:
\begin{mathpar}
    \inferrule{~}{%
        \ltr{b} \in \lang
    }
    \and
    \inferrule{%
        u, v \in \lang
    }{%
        \ltr{a} \cdot (u \parallel v) \in \lang
    }
\end{mathpar}

This language can be described by a pomset recogniser.
Let our carrier set be $M = \{ q_\ltr{a}, q_\ltr{b}, q_1, q_\bot, \unit \}$, and let $\odot$ and $\obar$ be the operations on $M$ given by
\begin{mathpar}
q \odot q' =
    \begin{cases}
    q & q' = \unit \\
    q' & q = \unit \\
    q_\ltr{b} & q = q_\ltr{a}, q' = q_1 \\
    q_\bot & \mathrm{otherwise}
    \end{cases}
\and
q \obar q' =
    \begin{cases}
    q & q' = \unit \\
    q' & q = \unit \\
    q_1 & q = q' = q_\ltr{b} \\
    q_\bot & \mathrm{otherwise}
    \end{cases}
\end{mathpar}
$\angl{M, \odot, \obar, \unit}$ is a bimonoid,  $F = \{ q_\ltr{b} \}$, and $i \colon \Sigma \to M$ is given by setting $i(\ltr{a}) = q_\ltr{a}$ and $i(\ltr{b}) = q_\ltr{b}$. 
One can then show that $\angl{M, \odot, \obar, \unit, i, F}$ accepts $\lang$.
\end{example}

\emph{Pomset contexts} are used to describe the behaviour of individual elements in a pomset recogniser.
Formally, the set of pomset contexts over a set $X$ is given by $\PC(X) = \SP(X \cup \{\cempty\})$.
Here the element $\cempty$ acts as a placeholder, where a pomset can be plugged in: given a context $c \in \PC(X)$ and $t \in \SP(X)$, let $\cinsert{c}{t} \in \SP(X)$ be obtained by substituting $t$ for $\cempty$ in $c$.

\section{Learning Pomset Recognisers}\label{sec:learning}

In this section we present our algorithm to learn pomset recognisers from an oracle (\emph{the teacher}) that answers \emph{membership} and \emph{equivalence queries}.
A membership query consists of a pomset, to which the teacher replies whether that pomset is in the language; an equivalence query consists of a \emph{hypothesis} pomset recogniser, to which the teacher replies \emph{yes} if it is correct or \emph{no} with a counterexample---a pomset incorrectly classified by the hypothesis---if it is not.

A pomset recogniser is essentially a tree automaton, with the additional constraint that its algebraic structure satisfies the bimonoid axioms.
Our algorithm is therefore relatively close to tree automata learning---in particular Drewes and H\"ogberg~\cite{DrewesH03,drewes2007}---but there are several key differences: we optimise the algorithm by taking advantage of the bimonoid axioms, and at the same time need to ensure that the hypotheses generated by the learning process satisfy those axioms.
%Below we first introduce the basic ingredients from which the algorithm is built.

\subsection{Observation Table}

We fix a target language $\lang \subseteq \SP$ throughout this section.
As in the original {\LStar} algorithm, the state of the learner throughout a run of the algorithm is given by a data structure called the \emph{observation table}, which collects information about $\lang$.
The table contains rows indexed by pomsets, representing the state reached by the correct pomset recogniser after reading that pomset; and columns indexed by pomset contexts, used to approximately indentify the behaviour of each state.
To represent the additional rows needed to approximate the pomset recogniser structure, we use the following definition.
Given $U \subseteq \SP$, we define
\[
	\ext{U} = \Sigma \cup \{u \cdot v : u, v \in U\} \cup \{u \parallel v : u, v \in U\} \subseteq \SP.
\]

\begin{definition}[Observation table]
	An \emph{observation table} is a pair $\angl{S, E}$, with $S \subseteq \SP$ subpomset-closed and $E \subseteq \PC$ such that $\pempty \in S$ and $\cempty \in E$.
	These sets induce the function $\row_{\angl{S, E}}
			\colon S \cup \ext{S} \to 2^E$:
%	\begin{align*}
	\(
			\row_{\angl{S, E}}(s)(e) = 1 
			\iff \cinsert{e}{s} \in \lang.
			\)
%	\end{align*}
	We often write $\row$ instead of $\row_{\angl{S, E}}$ when $S$ and $E$ are clear from the context.
\end{definition}

We depict observation tables, or more precisely $\row$, as two separate tables with rows in $S$ and $\ext{S} \setminus S$ respectively, see for instance \Cref{ex:nonclosed} below.

The goal of the learner is to extract a \emph{hypothesis} pomset recogniser from the rows in the table.
More specifically, the carrier of the underlying bimonoid of the hypothesis will be given by the rows indexed by pomsets in $S$.
The structure on the rows is obtained by transferring the structure of the row labels onto the rows (e.g., $\row(s) \odot \row(t) = \row(s \cdot t)$), but this is not well-defined unless the table satisfies \emph{closedness}, \emph{consistency}, and \emph{associativity}.
Closedness and consistency are standard in \LStar, whereas associativity is a new property specific to bimonoid learning.
We discuss each of these properties next, also including \emph{compatibility}, a property that is used to show minimality of
hypotheses.

The first potential issue is a closedness defect: this is the case when a composed row, indexed by an element of $\ext{S}$, is not indexed by a pomset in $S$.

\begin{example}[Table not closed]\label{ex:nonclosed}\footnotesize
	Recall $\lang = {\{ \ltr{a} \parallel \ltr{b} \}}^*$ from \Cref{ex:loop}, and suppose $S = \{\pempty,\ltr{a}, \ltr{b} \}$ and $E = \{\cempty, \ltr{a} \parallel \cempty, \cempty \parallel \ltr{b}\}$.
	The induced table is
	\begin{center}
		\begin{tabular}{L{3em} R{6ex} | C{3ex} C{6ex} C{6ex}}
			\multicolumn{5}{c}{$\hspace{14ex}\overbracket[.8pt][2pt]{\rule{15ex}{0pt}}^{\displaystyle E}$} \\
			& & \cempty & \ltr{a} \parallel \cempty & \cempty \parallel \ltr{b} \\
			\cline{2-5}
			\ldelim[{3}{8mm}[$S$\,] & \pempty & 1 & 0 & 0 \\
			& \ltr{a} & 0 & 0 & 1 \\
			& \ltr{b} & 0 & 1 & 0
		\end{tabular}
		\qquad
		\begin{tabular}{L{6em} R{6ex} | C{3ex} C{6ex} C{6ex}}
			\multicolumn{5}{c}{$\hspace{21.5ex}\overbracket[.8pt][2pt]{\rule{15ex}{0pt}}^{\displaystyle E}$} \\
			& & \cempty & \ltr{a} \parallel \cempty & \cempty \parallel \ltr{b} \\
			\cline{2-5}
			\ldelim[{7}{16mm}[$S^+ \setminus S$\,]
			& \ltr{aa} & 0 & 0 & 0 \\
			& \ltr{ab} & 0 & 0 & 0 \\
			& \ltr{ba} & 0 & 0 & 0 \\
			& \ltr{bb} & 0 & 0 & 0 \\
			& \ltr{a} \parallel \ltr{a} & 0 & 0 & 0 \\
			& \ltr{a} \parallel \ltr{b} & 1 & 0 & 0 \\
			& \ltr{b} \parallel \ltr{b} & 0 & 0 & 0
		\end{tabular}
	\end{center}
	The carrier of the hypothesis bimonoid is $M = \{\row(\pempty), \row(\ltr{a}), \row(\ltr{b})\}$, but the composition $\row(\ltr{a}) \odot \row(\ltr{a})$ cannot be defined since $\row(\ltr{aa}) \not\in M$.\end{example}

The absence of the issue described above is captured with \emph{closedness}.

\begin{definition}[Closed table]
	An observation table $\angl{S, E}$ is \emph{closed} if for all $t \in \ext{S}$ there exists $s \in S$ such that $\row(s) = \row(t)$.
\end{definition}

Another issue that may occur is that the same row being represented by different index pomsets leads to an inconsistent definition of the structure.
The absence of this issue is referred to as \emph{consistency}.

\begin{definition}[Consistent table]
	An observation table $\angl{S, E}$ is \emph{consistent} if for all $s_1, s_2 \in S$ such that $\row(s_1) = \row(s_2)$ we have for all $t \in S$ that
	\begin{mathpar}
		\row(s_1 \cdot t) = \row(s_2 \cdot t) \and
		\row(t \cdot s_1) = \row(t \cdot s_2) \and
		\row(s_1 \parallel t) = \row(s_2 \parallel t).
	\end{mathpar}
\end{definition}

Whenever closedness and consistency hold, one can define sequential and parallel composition operations on the rows of the table.
However, these operations are not guaranteed to be associative, as we show with the following example.

\begin{example}[Table not associative]\label{ex:nonassoc}
	Consider $\lang = \{\ltr{a}u : u \in {\{\ltr{b}\}}^*\}$ over $\Sigma = \{\ltr{a}, \ltr{b}\}$, and suppose $S = \{\pempty,\ltr{a},\ltr{b}\}$ and $E = \{\cempty, \cempty\ltr{a}\}$.
	The induced table is:
	\begin{center}\footnotesize
		\begin{tabular}{R{6ex} | C{3ex} C{5ex}}
			& \cempty & \cempty\ltr{a} \\
			\hline
			\pempty & 0 & 1\\
			\ltr{a} & 1 & 0\\
			\ltr{b} & 0 & 0
		\end{tabular}
		\qquad\qquad\qquad
		\begin{tabular}{R{6ex} | C{3ex} C{5ex}}
			& \cempty & \cempty\ltr{a} \\
			\hline
			\ltr{aa} & 0 & 0 \\
			\ltr{ab} & 1 & 0 \\
			\ltr{ba} & 0 & 0 \\
			\ltr{bb} & 0 & 0 \\
			\ltr{a} \parallel \ltr{a} & 0 & 0 \\
			\ltr{a} \parallel \ltr{b} & 0 & 0 \\
			\ltr{b} \parallel \ltr{b} & 0 & 0
		\end{tabular}
	\end{center}
	This table does not lead to an associative sequential operation on rows:
	\begin{align*}
		(\row(\ltr{a}) \odot \row(\ltr{b})) \odot \row(\ltr{a}) 
			= \row(\ltr{ab}) \odot \row(\ltr{a}) 
			= \row(\ltr{a}) \odot \row(\ltr{a}) 
			= \row(\ltr{aa}) \\
			\ne \row(\ltr{ab}) 
			= \row(\ltr{a}) \odot \row(\ltr{b}) 
			= \row(\ltr{a}) \odot \row(\ltr{ba}) 
			= \row(\ltr{a}) \odot (\row(\ltr{b}) \odot \row(\ltr{a})).
	\end{align*}
\end{example}

To prevent this issue we enforce the following additional property:

\begin{definition}[Associative table]\label{def:assoctable}
    Let $\asymbol \in \{{\cdot}, {\parallel}\}$.
	An observation table $\angl{S, E}$ is \emph{$\asymbol$-associative} if for all $s_1, s_2, s_3, s_l, s_r \in S$ with $\row(s_l) = \row(s_1 \asymbol s_2)$ and $\row(s_r) = \row(s_2 \asymbol s_3)$ we have $\row(s_l \asymbol s_3) = \row(s_1 \asymbol s_r)$.
	An observation table is \emph{associative} if it is both $\cdot$-associative and $\parallel$-associative.
\end{definition}

The table from \Cref{ex:nonassoc} is \emph{not} $\cdot$-associative: we have $\row(\ltr{a}) = \row(\ltr{ab})$ and $\row(\ltr{b}) = \row(\ltr{ba})$ but $\row(\ltr{aa}) \neq \row(\ltr{ab})$.

Putting the above definitions of closedness, consistency and associativity of tables together, we have the following result
for constructing a hypothesis.

\begin{lemma}[Hypothesis]\label{lem:hyp}
	A closed, consistent and associative table $\angl{S, E}$ induces a \emph{hypothesis} pomset recogniser $\hyp = \angl{H, \odot_H, \obar_H, \unit_H, i_H, F_H}$ where
	\begin{mathpar}
		H = \{\row(s) : s \in S\} \and
			\row(s_1) \odot_H \row(s_2) = \row(s_1 \cdot s_2) \and
			\row(s_1) \obar_H \row(s_2) = \row(s_1 \parallel s_2) \and
			\unit_H = \row(\pempty) \and
			i_H(\ltr{a}) = \row(\ltr{a}) \and
			F_H = \{\row(s) : s \in S, \row(s)(\cempty) = 1\}.
	\end{mathpar}
\end{lemma}
\begin{proof}
	The operations $\odot_H$ and $\obar_H$ are well-defined by closedness and consistency, and $\unit_H$ is well-defined because $\pempty \in S$ by the observation table definition.
	Commutativity of $\obar_H$ follows from commutativity of $\parallel$,
	and similarly that $\unit_H$ is a unit for both operations follows from $\pempty$
	being a unit. Associativity follows by associativity of the table (it does \emph{not} follow from $\cdot$ and $\parallel$ being associative:
	given elements $s_1, s_2, s_3 \in S$, $s_1 \cdot s_2 \cdot s_3$ is not necessarily present in $S \cup \ext{S}$).
	\qed%
\end{proof}

Since a hypothesis is constructed from an observation table $\angl{S, E}$ that records for given $s \in S$ and $e \in E$ whether $\cinsert{e}{s}$ is accepted by the language or not, one would expect that the hypothesis classifies those pomsets
\[
	T_{\angl{S, E}} = \{\cinsert{e}{s} : s \in S, e \in E\}
\]
correctly.
This is not necessarily the case, as we show in the following example.

\begin{example}
	Consider the language $\lang$ from \Cref{ex:nested}, and let $S = \{\pempty,\ltr{b}\}$ and $E = \{\cempty, \ltr{a}(\cempty \parallel \ltr{b})\}$.
	The induced table is
	\begin{center}
		\begin{tabular}{R{6ex} | C{3ex} C{10ex}}
			& \cempty & \ltr{a}(\cempty \parallel \ltr{b}) \\
			\hline
			\pempty & 0 & 0\\
			\ltr{b} & 1 & 1
		\end{tabular}
		\qquad\qquad\qquad
		\begin{tabular}{R{6ex} | C{3ex} C{10ex}}
			& \cempty & \ltr{a}(\cempty \parallel \ltr{b}) \\
			\hline
			\ltr{a} & 0 & 0\\
			\ltr{bb} & 0 & 0 \\
			\ltr{b} \parallel \ltr{b} & 0 & 0
		\end{tabular}
	\end{center}
	From this closed, consistent, and associative table we obtain a hypothesis pomset recogniser that satisfies
	\begin{align*}
		(\row(\ltr{a}) \odot (\row(\ltr{b}) \obar \row(\ltr{b})))(\cempty) &
			= (\row(\ltr{a}) \odot \row(\ltr{b} \parallel \ltr{b}))(\cempty) \\
					&			= (\row(\ltr{a}) \odot \row(\pempty))(\cempty)
			= \row(\ltr{a})(\cempty)			= 0 \ne 1
	\end{align*}
	and thus recognises a language that differs from $\lang$ on $\ltr{a \cdot (b \parallel b)} \in T_{\angl{S, E}}$.
\end{example}

We thus have the following definition, parametric in a subset of $T_{\angl{S, E}}$.

\begin{definition}[Compatible hypothesis]
	A closed, consistent, and associative observation table $\angl{S, E}$ induces a hypothesis $\hyp$ that is \emph{$X$-compatible} with its table, for $X \subseteq \SP$, if for $x \in X$ we have $x \in \lang_\hyp \iff x \in \lang$.
	We say that the hypothesis is \emph{compatible} with its table if it is $T_{\angl{S, E}}$-compatible with its table.
\end{definition}

Ensuring hypotheses are compatible with their table will not be a crucial step in proving termination, but plays a key role in ensuring minimality (\Cref{sec:minimality}). % and with that, as we will see in \Cref{sec:conversion}, in converting them from pomset recognisers to a special class of pomset automata.
This was originally shown by van Heerdt~\cite{heerdt2014} for Mealy machines.
\begin{figure}[ht!]
	\begin{minipage}{\textwidth}
		\begin{codebox}
			\li $S \gets \{\pempty\}$, $E \gets \{\cempty\}$
			\li \Repeat \label{line:loop1}
				\li \Repeat \label{line:loop2}
					\li \While $\angl{S, E}$ is not closed or not associative \label{line:loop3}
						\li \Do \If $\angl{S, E}$ is not closed
							\li\label{line:closedness} \Then find $t \in \ext{S}$ such that $\row(t) \neq \row(s)$ for all $s \in S$
							\li $S \gets S \cup \{t\}$ \label{line:addtoS}
						\End
						\li \For $\asymbol \in \{{\cdot}, {\parallel}\}$
							\li \Do \If $\angl{S, E}$ is not $\asymbol$-associative
								\li\label{line:sassoc} \Then find $s_1, s_2, s_3, s_l, s_r \in S$ and $e \in E$ such that
                                \zi\qquad $\row(s_l) = \row(s_1 \asymbol s_2)$,
								\zi\qquad $\row(s_r) = \row(s_2 \asymbol s_3)$, and
                                \zi\qquad $\row(s_l \asymbol s_3)(e) \ne \row(s_1 \asymbol s_r)(e)$
								\li let $b$ be the result of a membership query on $s_1 \asymbol s_2 \asymbol s_3$
								\li \If $\row(s_l \asymbol s_3)(e) \ne b$ \Then
									\li $E \gets E \cup \{\cinsert{e}{\cempty \asymbol s_3}\}$
								\li \Else
									\li $E \gets E \cup \{\cinsert{e}{s_1 \asymbol \cempty}\}$
								\End
							\End
						\End
					\End
					\li\label{line:hypothesis}construct the hypothesis $\hyp$ for $\angl{S, E}$
					\li \If $\hyp$ is not compatible with its table
						\li \Then find $s \in S$ and $e \in E$ such that $\cinsert{e}{s} \in \lang_\hyp \iff \cinsert{e}{s} \not\in \lang$
						\li $E \gets E \cup \{\proc{HandleCounterexample}(S, E, \cinsert{e}{s}, \cempty)\}$
					\End
				\li \Until $\hyp$ is compatible with its table
				\li \If the teacher replies \emph{no} to $\hyp$, with a counterexample $z$
					\li \Then $E \gets E \cup \{\proc{HandleCounterexample}(S, E, z, \cempty)\}$
				\End
			\li \Until the teacher replies \emph{yes}
			\li \Return $\hyp$
		\end{codebox}
		\begin{codebox}
			\Procname{$\proc{HandleCounterexample}(S, E, z, c)$}
			\li \If $z \in S \cup \ext{S}$
				\li \Then let $s \in S$ be such that $\row(s) = \row(z)$
				\li \If $\cinsert{c}{s} \in \lang \iff \cinsert{c}{z} \in \lang$
					\li \Then \Return $s$
				\li \Else
					\li \Return $c$
				\End
			\End
			\li let non-empty $u_1, u_2 \in \SP$ and $\asymbol \in \{{\cdot}, {\parallel}\}$ be such that $u_1 \asymbol u_2 = z$
			\li $u_1 \gets \proc{HandleCounterexample}(S, E, u_1, \cinsert{c}{\cempty \asymbol u_2})$
			\li \If $u_1 \not\in S$
				\li \Then \Return $u_1$
			\End
			\li $u_2 \gets \proc{HandleCounterexample}(S, E, u_2, \cinsert{c}{u_1 \asymbol \cempty})$
			\li \If $u_2 \not\in S$
				\li \Then \Return $u_2$
			\End
			\li \Return $\proc{HandleCounterexample}(S, E, u_1 \asymbol u_2, c)$
		\end{codebox}
	\end{minipage}
    \captionof{algorithm}{The pomset recogniser learning algorithm.}\label{fig:plstar}
\end{figure}

\subsection{The Learning Algorithm}

We are now ready to introduce our learning algorithm, \Cref{fig:plstar}.
%It uses a separate procedure called $\proc{HandleCounterexample}$ to handle counterexamples.
The main algorithm initialises the table to $\angl{\{\pempty\}, \{\cempty\}}$ and starts by augmenting the table to make sure it is closed and associative.
We give an example below.

\begin{example}[Fixing closedness and associativity]\label{ex:fixes}
	Consider the table from \Cref{ex:nonclosed}, where $\row(\ltr{aa}) \not\in \{\row(\pempty), \row(\ltr{a}), \row(\ltr{b})\}$ witnesses a closedness defect.
	To fix this, the algorithm would add $\ltr{aa}$ to the set $S$, which means $\row(\ltr{aa})$ will become part of the carrier of the hypothesis.

	Now consider the table from \Cref{ex:nonassoc}.
	Here we found an associativity defect witnessed by $\row(\ltr{a}) = \row(\ltr{ab})$ and $\row(\ltr{b}) = \row(\ltr{ba})$ but $\row(\ltr{aa}) \neq \row(\ltr{ab})$.
	More specifically, $\row(\ltr{aa})(\cempty) \neq \row(\ltr{ab})(\cempty)$.
	Thus, $s_1 = s_3 = s_l = \ltr{a}$, $s_2 = s_r = \ltr{b}$, $s_l = \ltr{a}$, and $e = \cempty$.
	A membership query on $\ltr{aba}$ shows $\ltr{aba} \not\in \lang$, so $b = 0$.
	We have $\row(\ltr{aa})(\cempty) = 0$, and therefore the algorithm would add the context $\cinsert{\cempty}{\ltr{a} \cdot \cempty} = \ltr{a} \cdot \cempty$ to $E$.
\end{example}
Note that the algorithm does not explicitly check for consistency; this is because we actually ensure a stronger property---sharpness~\cite{barlocco2018}---as an invariant (\Cref{lem:nonce}). This property ensures every row indexed by a pomset in $S$ is indexed by exactly one pomset in $S$ (implying consistency):
\begin{definition}[Sharp table]
	An observation table $\angl{S, E}$ is \emph{sharp} if for all $s_1, s_2 \in S$ such that $\row(s_1) = \row(s_2)$ we have $s_1 = s_2$.
\end{definition}

The idea of maintaining sharpness is due to Maler and Pnueli~\cite{maler1995}.

Once the table is closed and associative, we construct the hypothesis and check if it is compatible with its table.
If this is not the case, a witness for incompatibility is a counterexample by definition, so $\proc{HandleCounterexample}$ is invoked to extract an extension of $E$, and we return to checking closedness and associativity.
Once we obtain a hypothesis that is compatible with its table, we submit it to the teacher to check for equivalence with the target language.
If the teacher provides a counterexample, we again process this and return to checking closedness and associativity.
Once we have a compatible hypothesis for which there is no counterexample, we return this correct pomset recogniser.

The procedure $\proc{HandleCounterexample}$, adapted from~\cite{DrewesH03,drewes2007}, is provided with an observation table $\angl{S, E}$ a pomset $z$, and a context $c$ and finds a single context to add to $E$.
The main invariant is that $\cinsert{c}{z}$ is a counterexample.
Recursive calls replace subpomsets from $\ext{S}$ with elements of $S$ in this counterexample while maintaining the invariant.
There are two types of return values: if $c$ is a suitable context, $c$ is returned; otherwise the return value is an element of $S$ that is to replace $z$.
The context $c$ is suitable if $z \in \ext{S}$ and adding $c$ to $E$ would distinguish $\row(s)$ from $\row(z)$, where $s \in S$ is such that currently $\row(s) = \row(z)$.
Because $S$ is non-empty and subpomset-closed, if $z \not\in S \cup \ext{S}$ it can be decomposed into $z = u_1 \asymbol u_2$ for non-empty $u_1, u_2 \in \SP$ and $\asymbol \in \{{\cdot}, {\parallel}\}$.
We then recurse into $u_1$ and $u_2$ to replace them with elements of $S$ and replace $z$ with $u_1 \asymbol u_2 \in \ext{S}$ in a final recursive call.
If $c = \cempty$, the return value cannot be in $S$, as we will show in \Cref{lem:nonce} that these elements are not counterexamples.

\begin{example}[Processing a counterexample]\label{ex:counterexample}
	Consider $\lang = \{\ltr{a}, \ltr{aa}, \ltr{a} \parallel \ltr{a}\}$, and let $S = \{\pempty, \ltr{a}\}$ and $E = \{\cempty\}$.
	This induces a closed, sharp, and associative table
	\begin{center}
		\begin{tabular}{R{6ex} | C{3ex}}
			& \cempty \\
			\hline
			\pempty & 0 \\
			\ltr{a} & 1
		\end{tabular}
		\qquad\qquad\qquad
		\begin{tabular}{R{6ex} | C{3ex}}
			& \cempty \\
			\hline
			\ltr{aa} & 1 \\
			\ltr{a} \parallel \ltr{a} & 1
		\end{tabular}
	\end{center}
	Suppose an equivalence query on its pomset recogniser, which rejects only the empty pomset, gives counterexample $z = \ltr{a} \parallel \ltr{a} \parallel \ltr{aa}$.
	We may decompose $z$ as $\cinsert{(\cempty \parallel \ltr{aa})}{\ltr{a} \parallel \ltr{a}}$, where $\ltr{a} \parallel \ltr{a} \in \ext{S} \setminus S$.
	Because $\row(\ltr{a} \parallel \ltr{a}) = \row(\ltr{a})$, $\cinsert{(\cempty \parallel \ltr{aa})}{\ltr{a}} = \ltr{a} \parallel \ltr{aa}$, and $\ltr{a} \parallel \ltr{aa} \in \lang \iff z \in \lang$, we update $z = \ltr{a} \parallel \ltr{aa}$ and repeat the process.
	Now we decompose $z = \cinsert{(\ltr{a} \parallel \cempty)}{\ltr{aa}}$.
	Since $\row(\ltr{aa}) = \row(\ltr{a})$, $\cinsert{(\ltr{a} \parallel \cempty)}{\ltr{a}} = \ltr{a} \parallel \ltr{a}$, and $\ltr{a} \parallel \ltr{a} \in \lang \iff z \not\in \lang$, we finish by adding $\ltr{a} \parallel \cempty$ to $E$.
\end{example}

\subsection{Termination and Query Complexity}

Our termination argument is based on a comparison of the current observation table with the infinite table $\angl{\SP, \PC}$.
We first show that the latter induces a hypothesis, called the \emph{canonical pomset recogniser} for the language.
Its underlying bimonoid is isomorphic to the syntactic bimonoid~\cite{lodaya-weil-2000} for the language.

\begin{lemma}\label{lem:fullhyp}
	$\angl{\SP, \PC}$ is a closed, consistent, and associative observation table.
\end{lemma}

\begin{definition}[Canonical pomset recogniser]
	The \emph{canonical pomset re\-cogniser} for $\lang$ is the the hypothesis for the observation table $\angl{\SP, \PC}$.
	We denote this hypothesis by $\angl{M_\lang, \odot_\lang, \obar_\lang, \unit_\lang, i_\lang, F_\lang}$.
\end{definition}

The comparison of the current table with $\angl{\SP, \PC}$ is in terms of the number of distinct rows they hold.
In the following lemma we show that the number of the former is bounded by the number of the latter.

\begin{lemma}\label{lem:bound}
	If $M_\lang$ is finite, any observation table $\angl{S, E}$ satisfies
	\[
		|\{\row(s) : s \in S\}| \le |M_\lang|.
	\]
\end{lemma}
\begin{proof}
	Note that $M_\lang = \{\row_{\angl{\SP, \PC}}(s) : s \in S\}$.
	Given $s_1, s_2 \in S$ such that $\row_{\angl{S, E}}(s_1) \ne \row_{\angl{S, E}}(s_2)$ we have $\row_{\angl{\SP, \PC}}(s_1) \ne \row_{\angl{\SP, \PC}}(s_2)$.
	This implies $|\{\row(s) : s \in S\}| \le |M_\lang|$.
	\qed%
\end{proof}

An important fact will be that none of the pomsets in $S$ can form a counterexample for the hypothesis of a table $\angl{S, E}$.
In order to show this we will first show that the hypothesis is always \emph{reachable}, a concept we define for arbitrary pomset recognisers below.

\begin{definition}[Reachability]
	A pomset recogniser $\mathcal{R} = \angl{M, \odot, \obar, \unit, i, F}$ is \emph{reachable} if for all $m \in M$ there exists $u \in \SP$ such that $\free{i}(u) = m$.
\end{definition}

Our reachability lemma relies on the fact that $S$ is subpomset-closed.

\begin{lemma}[Hypothesis reachability]\label{lem:reachability}
	Given a closed, consistent, and associative observation table $\angl{S, E}$, the hypothesis it induces is reachable.
	In particular, $\free{i_H}(s) = \row(s)$ for any $s \in S$.
\end{lemma}

From the above it follows that we always have compatibility with respect to the set of row indices, as we show next.

\begin{lemma}\label{lem:nonce}
	The hypothesis of any closed, consistent, and associative observation table $\angl{S, E}$ is $S$-compatible.
\end{lemma}

Before turning to our termination proof, we show that some simple properties hold throughout a run of the algorithm.

\begin{lemma}[Invariant]\label{lm:invariant}
	Throughout execution of \Cref{fig:plstar}, we have that $\angl{S, E}$ is a sharp observation table.
\end{lemma}
\begin{proof}
	Subpomset-closedness holds throughout each run since $\{\pempty\}$ is subpomset-closed and adding a single element of $\ext{S}$ to $S$ preserves the property.

	For sharpness, first note that the initial table is sharp as it only has one row.
	Sharpness of $\angl{S, E}$ can only be violated when adding elements to $S$. But
	the only place where this happens is on line~\ref{line:addtoS}, and there
	the new row is unequal to all previous rows, which means sharpness is preserved.
	\qed%
\end{proof}

The preceding results allow us to prove our termination theorem.

\begin{theorem}[Termination]
	If $M_\lang$ is finite, then \Cref{fig:plstar} terminates.
\end{theorem}
\begin{proof}
	First, we observe that fixing a closedness defect by adding a row (line~\ref{line:addtoS}) can only happen finitely many times, since, by \Cref{lem:bound}, the size of $\{\row(s) : s \in S\}$ is bounded by $M_\lang$.

	This means that it suffices to show the following two points:
	\begin{enumerate}
		\item Each iteration of any of the loops starting on lines~\ref{line:loop1}--\ref{line:loop3} either fixes a closedness defect by adding a row, or adapts $E$ so that $\angl{S,E}$ ends up \emph{not} being closed at the end of loop body. In the second case, a closedness defect will be fixed in the following iteration of the inner while loop.
		\item The calls to $\proc{HandleCounterexample}$ terminate.
	\end{enumerate}
	Combined, these show that the algorithm terminates.
	For the first point, we treat each of the cases:
	\begin{itemize}
	\item
		If the table is not closed, we directly find a new row that is taken from the $\ext{S}$-part of the table and added to the $S$-part of the table.
		\item
			Consider the failure of $\asymbol$-associativity, for $\asymbol \in \{{\cdot}, {\parallel}\}$, and let $s_1, s_2, s_3, s_l, s_r \in S$ and $e \in E$ be such that $\row(s_l) = \row(s_1 \asymbol s_2)$, $\row(s_r) = \row(s_2 \asymbol s_3)$, and $\row(s_l \asymbol s_3)(e) \ne \row(s_1 \asymbol s_r)(e)$.
			Suppose $\row(s_l \asymbol s_3)(e) \ne b$, with $b$ be the result of a membership query on $s_1 \asymbol s_2 \asymbol s_3$.
			Then $\cinsert{e}{\cempty \asymbol s_3}$ distinguishes the previously equal rows $\row(s_1 \asymbol s_2)$ and $\row(s_l)$, so adding it to $E$ creates a closedness defect.
			The fact that $\row(s_1 \asymbol s_2)$ cannot remain equal to another row than $\row(s_l)$ is a result of the sharpness invariant.

			Alternatively, $\row(s_l \asymbol s_3)(e) = b$ means $\row(s_1 \asymbol s_r)(e) \ne b$, for otherwise we would contradict $\row(s_l \asymbol s_3)(e) \ne \row(s_1 \asymbol s_r)(e)$.
			For similar reasons the context $\cinsert{e}{s_1 \asymbol \cempty}$ in this case distinguishes the previously equal rows $\row(s_1 \asymbol s_2)$ and $\row(s_r)$, creating a closedness defect.
		\item
			A compatibility defect results in the identification of a counterexample, the handling of which we discuss next.
		\item
			Whenever a counterexample is identified, we eventually find a context $c$, $s \in S$, and $t \in \ext{S} \setminus S$ such that $\row(t) = \row(s)$ and $\cinsert{c}{t} \in \lang \iff \cinsert{c}{s} \not\in \lang$.
			Thus, adding $c$ to $E$ creates a closedness defect.
	\end{itemize}

Termination of $\proc{HandleCounterexample}$ follows: the first two recursive calls in the procedure replace $z$ with strict subpomsets of $z$, whereas the last one replaces $z$ with an element of $\ext{S}$, so no further recursion will happen.
	\qed%
\end{proof}

\paragraph{Query Complexity.}
We determine upper bounds on the membership and equivalence query numbers of a run of the algorithm in terms of the size of the canonical pomset recogniser $n = |M_\lang|$, the size of the alphabet $k = |\Sigma|$, and the maximum number of operations (from $\{\cdot, \parallel\}$, used to compose alphabet symbols) $m$ found in a counterexample.
We note that since the number of distinct rows indexed by $S$ is bounded by $n$ and the table remains sharp throughout any run, the final size of $S$ is at most $n$.
Thus, the final size of $\ext{S}$ is in $\bigOh{n^2 + k}$.
Given the initialisation of $S$ with a single element, the number of closedness defects fixed throughout a run is at most $n - 1$.
This means that the total number of associativity defects fixed and counterexamples handled (including those resulting from compatibility defects) together is $n - 1$.
We can already conclude that the number of equivalence queries posed is bounded by $n$.
Moreover, we know that the final table will have at most $n$ columns, and therefore the total number of cells in that table will be in $\bigOh{n^3 + kn}$.

The number of membership queries posed during a run of the algorithm is given by the number of cells in the table plus the number of queries needed during the processing of counterexamples.
Consider the counterexample $z$ that contains the maximum number of operations among those encountered during a run.
The first two recursive calls of $\proc{HandleCounterexample}$ break down one operation, whereas the third is used to execute a base case making two membership queries and does not lead to any further recursion.
The number of membership queries made starting from a given counterexample is thus in $\bigOh{m}$.
This means the total number of membership queries during the processing of counterexamples is in $\bigOh{mn}$, from which we conclude that the number of membership queries posed during a run is in $\bigOh{n^3 + mn + kn}$.

\subsection{Minimality of Hypotheses}\label{sec:minimality}

In this section we will show that all hypotheses submitted by the algorithm to the teacher are minimal.
We first need to define what minimality means.
As is the case for DFAs, it is the combination of an absence of unreachable states and of every state exhibiting its own distinct behaviour.

\begin{definition}[Minimality]
	A pomset recogniser $\pr = \angl{M, \odot, \obar, \unit, i, F}$ is \emph{minimal} if it is reachable and for all $u, v \in \SP$ with $\free{i}(u) \ne \free{i}(v)$ there exists $c \in \PC$ such that $\cinsert{c}{u} \in \lang_\pr \iff \cinsert{c}{v} \not\in \lang_\pr$.
\end{definition}

Before proving the main result of this section, we need the following:

\begin{lemma}\label{lem:sound}
	For all pomset recognisers $\angl{M, \odot, \obar, \unit, i, F}$ and $u, v \in \SP$ such that $\free{i}(u) = \free{i}(v)$ we have for any $c \in \PC$ that $\free{i}(\cinsert{c}{u}) = \free{i}(\cinsert{c}{v})$.
\end{lemma}

The minimality theorem below relies on table compatibility, which allows us to distinguish the behaviour of states based on the contents of their rows.
Note that the algorithm only submits a hypothesis in an equivalence query if that hypothesis is compatible with its table.

\begin{theorem}[Minimality of hypotheses]\label{prop:minimality}
	A closed, consistent, and associative observation $\angl{S, E}$ induces a minimal hypothesis if the hypothesis is compatible with its table.
\end{theorem}
\begin{proof}
	We obtain the hypothesis from \Cref{lem:hyp}.
	Since $S$ is subpomset-closed, we have by \Cref{lem:reachability} that the hypothesis is reachable.
	Moreover, for every $s \in S$ we have $\free{i_H}(s) = \row(s)$. 	Consider $u_1, u_2 \in \SP$ such that $\free{i_H}(u_1) \ne \free{i_H}(u_2)$.
	Then there exist $s_1, s_2 \in S$ such that $\row(s_1) = \free{i_H}(u_1)$ and $\row(s_2) = \free{i_H}(u_2)$, and we have $\row(s_1) \ne \row(s_2)$.
	Let $e \in E$ be such that $\row(s_1)(e) \ne \row(s_2)(e)$.
	We have
	\begin{align*}
		\free{i_H}(\cinsert{e}{u_1}) \in F_H &
			\iff \free{i_H}(\cinsert{e}{s_1}) \in F_H
			\tag*{(\Cref{lem:sound}) \phantom{\qed}}\\
		&
			\iff \cinsert{e}{s_1} \in \lang_\hyp \\
		&
			\iff \row(s_1)(e) = 1 \\
		&
			\iff \row(s_2)(e) = 0 \\
		&
			\iff \cinsert{e}{s_2} \not\in \lang_\hyp \\
		&
			\iff \free{i_H}(\cinsert{e}{s_2}) \not\in F_H \\
		&
			\iff \free{i_H}(\cinsert{e}{u_2}) \not\in F_H.
			\tag*{(\Cref{lem:sound}) \qed}
	\end{align*}
\end{proof}

As a corollary, we find that the canonical pomset recogniser is minimal.

\begin{proposition}\label{prop:canmin}
	The canonical pomset recogniser is minimal.
\end{proposition}

\section{Conversion to Pomset Automata}\label{sec:conversion}

Bimonoids are a useful representation of pomset languages because sequential and parallel composition are on an equal footing; in the case of the learning algorithm of the previous section, this helps us treat both operations similarly.
On the other hand, the behaviour of a program is usually thought of as a series of actions, some of which involve launching two or more threads that later combine.
Here, sequential actions form the basic unit of computation, while fork/join patterns of threads are specified separately.
\emph{Pomset automata}~\cite{kappe-brunet-luttik-silva-zanasi-2018b} encode this more asymmetric model: they can be thought of as non-deterministic finite automata with an additional transition type that brokers forking and joining threads.

In this section, we show how to convert a pomset recogniser to a certain type of pomset automaton, where acceptance of a pomset is guided by its structure; conversely, we show that each of the pomset automata in this class can be represented by a pomset recogniser.
Together with the previous section, this establishes that the languages of pomset automata in this class are learnable.

If $S$ is a set, we write $\M(S)$ for the set of \emph{finite multisets} over $S$.
A finite multiset over $S$ is written $\phi = \mset{s_1, \dots, s_n}$.

\begin{definition}[Pomset automata]
A \emph{pomset automaton} (PA) is a tuple $A =\angl{Q, I, F, \delta, \gamma}$ where
\begin{itemize}
    \item $Q$ is a set of \emph{states}, with $I, F \subseteq Q$ the \emph{initial} and \emph{accepting states}, and
    \item $\delta \colon Q \times \Sigma \to 2^Q$ the \emph{sequential transition function}, and
    \item $\gamma \colon Q \times \M(Q) \to 2^Q$ the \emph{parallel transition function}.
\end{itemize}
Lastly, for every $q \in Q$ there are finitely many $\phi \in \M(Q)$ such that $\gamma(q, \phi) \neq \emptyset$.
\end{definition}

A finite PA can be represented graphically: every state is drawn as a vertex, with accepting states doubly circled and initial states pointed out by an arrow, while $\delta$-transitions are represented by labelled edges, and $\gamma$-transitions are drawn as a multi-ended edge.
For instance, in \Cref{figure:simple-pa}, we have drawn a PA with states $q_0$ through $q_5$ with $q_5$ accepting, and $q_1 \in \delta(q_0, \ltr{a})$ (among other $\delta$-transitions), while the multi-ended edge represents that $q_2 \in \gamma(q_1, \mset{q_3, q_4})$, i.e., $q_2$ can launch threads starting in $q_3$ and $q_4$, which, upon termination, resume in $q_2$.%

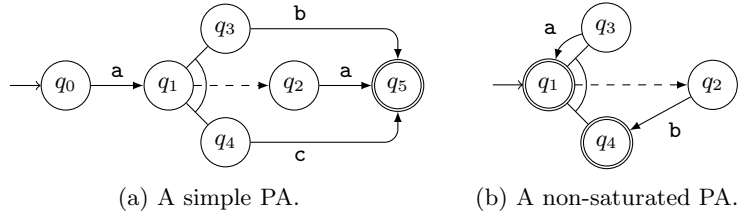
\begin{figure}
    \centering
    \begin{subfigure}[t]{0.55\textwidth}
        \centering
        \begin{tikzpicture}[every node/.style={transform shape},initial text={}]
            \node[state,initial] (q0) {$q_0$};
            \node[state,right=7mm of q0] (q1) {$q_1$};
            \node[state,above right=4mm of q1] (q3) {$q_3$};
            \node[state,below right=4mm of q1] (q4) {$q_4$};
            \node[state,right=10mm of q1] (q2) {$q_2$};
            \node[state,accepting,right=7mm of q2] (q5) {$q_5$};

            \draw (q0) edge[-latex] node[above] {$\ltr{a}$} (q1);
            \draw[-] (q1) edge (q3);
            \draw[-] (q1) edge (q4);
            \draw[dashed] (q1) edge[-latex] (q2);
            \draw[rounded corners=5pt,-latex] (q3) -| node[above,xshift=-4em] {$\ltr{b}$} (q5);
            \draw[rounded corners=5pt,-latex] (q4) -| node[below,xshift=-4em] {$\ltr{c}$} (q5);
            \draw pic[draw,angle radius=.5cm] {angle=q4--q1--q3};
            \draw (q2) edge[-latex] node[above] {$\ltr{a}$} (q5);
        \end{tikzpicture}
        \caption{A simple PA.}%
        \label{figure:simple-pa}
    \end{subfigure}
    \begin{subfigure}[t]{0.3\textwidth}
        \centering
        \begin{tikzpicture}[scale=2,initial text={},initial distance=2mm]
            \node[state,accepting,initial] (q1) {$q_1$};
            \node[state,right=15mm of q1] (q2) {$q_2$};
            \node[state,accepting,below right=4mm of q1] (q4) {$q_4$};
            \node[state,above right=4mm of q1] (q3) {$q_3$};

            \draw[-] (q1) edge (q3);
            \draw[-] (q1) edge (q4);
            \draw pic[draw,angle radius=.5cm] {angle=q4--q1--q3};
            \draw[dashed,-latex] (q1) edge (q2);
            \draw[-latex] (q2) edge node[below right] {$\ltr{b}$} (q4);
            \draw[-latex] (q3) edge[bend right] node[above left] {$\ltr{a}$} (q1);
        \end{tikzpicture}
        \caption{A non-saturated PA.}%
        \label{figure:problematic-pa}
    \end{subfigure}
    \caption{Some pomset automata.}%
    \label{figure:example-pas}
\end{figure}

The sequential transition function is interpreted as in non-deterministic finite automata: if $q' \in \delta(q, \ltr{a})$, then a machine in state $q$ may transition to state $q'$ after performing the action $\ltr{a}$.
The intuition to the parallel transition function is that if $q' \in \gamma(q, \mset{r_1, \dots, r_n})$, then a machine in state $q$ may launch threads starting in states $r_1$ through $r_n$, and when each of those has terminated succesfully, may proceed in state $q'$.
Note how the representation of starting states in a $\gamma$-transition allows for the possibility of launching multiple instances of the same thread, and disregards their order---i.e., $\gamma(q, \mset{r_1, \dots, r_n}) = \gamma(q, \mset{r_n, \dots, r_1})$.
This intuition is made precise through the notion of a \emph{run}.

\begin{definition}[Run relation]%
\label{definition:runs}
The \emph{run relation} of a PA $A = \angl{Q, I, F, \delta, \gamma}$, denoted ${\runrel_A}$, is defined as the the smallest subset of $Q \times \SP \times Q$ satisfying
\begin{mathpar}
\inferrule{~}{%
    q \arun{\pempty}_A q
}
\and
\inferrule{
    q' \in \delta(q, \ltr{a})
}{%
    q \arun{\ltr{a}}_A q'
}
\and
\inferrule{%
    \forall 1 \leq i \leq n.\ r_i \arun{u_i}_A r_i' \in F \\\\
    q' \in \gamma(q, \mset{r_1, \dots, r_n})
}{%
    q \arun{u_1 \parallel \dots \parallel u_n}_A q'
}
\and
\inferrule{%
    q \arun{u}_A q'' \\\\
    q'' \arun{v}_A q'
}{%
    q \arun{u \cdot v}_A q'
}
\end{mathpar}
The \emph{language accepted} by $A$ is
\(
    \lang_A = \{ u \in \SP : \exists q \in I, q' \in F.\ q \arun{u}_A q' \}
\).
\end{definition}

\begin{example}
If $A$ is the PA from \Cref{figure:simple-pa}, we can see that $q_3 \arun{\ltr{b}}_A q_5$ and $q_4 \arun{\ltr{c}}_A q_5$ as a result of the second rule; by the third rule, we find that $q_1 \arun{\ltr{b} \parallel \ltr{c}}_A q_2$.
Since $q_2 \arun{\ltr{a}} q_5$ and $q_0 \arun{a}_A q_1$ (again by the second rule), we can conclude $q_0 \arun{\ltr{a} \cdot (\ltr{b} \parallel \ltr{c}) \cdot \ltr{a}}_A q_5$ by repeated application of the last rule.
The language accepted by this PA is the singleton set $\{ \ltr{a} \cdot (\ltr{b} \parallel \ltr{c}) \cdot \ltr{a} \}$.
\end{example}

In general, finite pomset automata can accept a very wide range of pomset languages, including all context free (pomset) languages~\cite{kappe-brunet-luttik-silva-zanasi-2019}.
The intuition behind this is that the mechanism of forking and joining encoded in $\gamma$ can be used to simulate a call stack.
For example, the automaton in \Cref{figure:problematic-pa} accepts the strictly context-free language (of words) $\{ \ltr{a}^n \cdot \ltr{b}^n : n \in \naturals \}$.
It follows that PAs can represent strictly more pomset languages than pomset recognisers.
To tame the expressive power of PAs at least slightly, we propose the following.

\begin{definition}[Saturation]
We say that $A = \angl{Q, I, F, \delta, \gamma}$ is \emph{saturated} when for all $u, v \in \SP$ with $u, v \neq \pempty$, both of the following are true:
\begin{enumerate}[label={(\roman*)}]
    \item
    If $q \arun{u \cdot v}_A q'$, then there exists a $q'' \in Q$ with $q \arun{u}_A q''$ and $q'' \arun{v}_A q'$.

    \item
    If $q \arun{u \parallel v}_A q'$, then there exist $r, s \in Q$ and $r', s' \in F$ such that
    \begin{mathpar}
    r \arun{u}_A r' \and
    s \arun{v}_A s' \and
    q' \in \gamma(q, \mset{r, s})
    \end{mathpar}
\end{enumerate}
\end{definition}

\begin{example}
Returning to \Cref{figure:example-pas}, we see that the PA in \Cref{figure:simple-pa} is saturated, while \Cref{figure:problematic-pa} is not, as a result of the run $q_1 \arun{\ltr{a} \cdot \ltr{a} \cdot \ltr{b} \cdot \ltr{b}}_A q_4$, which does not admit an intermediate state $q$ such that $q_1 \arun{\ltr{a} \cdot \ltr{a}}_A q$ and $q \arun{\ltr{b} \cdot \ltr{b}}_A q_4$.
\end{example}

We now have everything in place to convert the encoding of a language given by a pomset recogniser to a pomset automaton.
The idea is to represent every element $q$ of the bimonoid by a state which accepts exactly the language of pomsets mapped to $q$; the transition structure is derived from the operations.

\begin{lemma}%
\label{lemma:recognisable-to-saturated-regular}
Let $\mathcal{R} = \angl{M, \odot, \obar, \unit, i, F}$ be a pomset recogniser.
We construct the pomset automaton $A = \angl{M, F, \{ \unit \}, \delta, \gamma}$ (note: we use $F$ as the set of initial states) where $\delta \colon M \times \Sigma \to 2^M$ and $\gamma \colon M \times \M(M) \to 2^M$ are given by
\begin{mathpar}
\delta(q, \ltr{a}) = \{ q' : i(\ltr{a}) \odot q' = q \}
\and
\gamma(q, \phi) = \{ q' : (r \obar r') \odot q' = q,\ \phi = \mset{r, r'} \}
\end{mathpar}
Then $A$ is saturated, and $\lang_A = \lang_\mathcal{R}$.
\end{lemma}

\begin{example}
Let $\angl{M, \odot, \obar, \unit, i, F}$ be the pomset recogniser from \Cref{ex:nested}.
The pomset automaton that arises from the construction above is partially depicted in \Cref{figure:recognisable-to-saturated-regular}; we have not drawn the state $q_\bot$ and its incoming transitions, or forks into $\unit$, to avoid clutter.
In this PA, we see that, since $q_\ltr{a} \odot q_1 = q_\ltr{b}$ and $i(\ltr{a}) = q_\ltr{a}$, we have $q_1 \in \delta(q_\ltr{b}, \ltr{a})$.
Furthermore, since $(q_\ltr{b} \obar q_\ltr{b}) \odot \unit = q_1 \odot \unit = q_1$, we also have $\unit \in \gamma(q_1, \mset{q_\ltr{b}, q_\ltr{b}})$.
Finally, $q_\ltr{b}$ is initial, since $F = \{ q_\ltr{b} \}$.
\end{example}

\begin{figure}
    \centering
    \begin{tikzpicture}[initial text={}]
        \node[state] (qb) {$q_\ltr{b}$};
        \node[state,accepting,right=of qb] (unit) {$\unit$};
        \node[state,left=of qb] (q1a) {$q_1$};
        \node[state,right=of unit] (qa) {$q_\ltr{a}$};

        \draw[->] (qb) edge node[above] {$\ltr{b}$} (unit);
        \draw[->] (qb) edge node[above] {$\ltr{a}$} (q1a);
        \draw (q1a.north east) edge[bend left] (qb.north west);
        \draw (q1a.south east) edge[bend right] (qb.south west);
        \draw (q1a) + (38:5mm) arc (38:322:5mm);
        \path (q1a.south) edge[bend right,dashed,-latex,in=-90,out=-90,looseness=0.4] (unit.south);
        \path (qa) edge[->] node[above] {$\ltr{a}$} (unit);
    \end{tikzpicture}
    \caption{%
        Part of the PA obtained from the pomset recogniser from \Cref{ex:nested}, using the construction from \Cref{lemma:recognisable-to-saturated-regular}.
        The state $q_\bot$ (which does not contribute to the language of the automaton) and forks into the state $\unit$ are not pictured.
    }%
    \label{figure:recognisable-to-saturated-regular}
\end{figure}
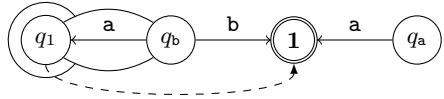

We have thus shown that the language of any pomset recogniser can be accepted by a finite and saturated PA\@.
In turn, this shows that our algorithm can, in principle, be adapted to work with a teacher that takes a (saturated) PA instead of a pomset recogniser as hypothesis, by simply converting the hypothesis pomset recogniser to an equivalent PA before sending it over.

Conversely, we can show that the transition relations of a saturated PA carry the algebraic structure of a bimonoid, and use that to show that a language recognised by a saturated PA is also recognised by a bimonoid.
This shows that our characterisation is ``tight'', i.e., languages recognised by saturated PAs are precisely those recognised by bimonoids, and hence learnable.

\begin{lemma}%
\label{lemma:saturated-regular-to-recognisable}
Let $A = \angl{Q, I, F, \delta, \gamma}$ be a saturated pomset automaton.
We can construct a pomset recogniser $\mathcal{R} = \angl{M, \odot, \obar, \unit, i, F'}$, where
\begin{mathpar}
M = \{ {\arun{u}_A} : u \in \SP \}
\and
{\arun{u}_A} \odot {\arun{v}_A} = {\arun{u \cdot v}_A}
\and
{\arun{u}_A} \obar {\arun{v}_A} = {\arun{u \parallel v}_A}
\and
i(\ltr{a}) = {\arun{\ltr{a}}_A}
\and
F' = \{ {\arun{u}_A} \in M : \exists q \in I, q' \in F.\ q \arun{u}_A q' \}
\end{mathpar}
Now $\odot$ and $\obar$ are well-defined, and $\mathcal{R}$ is a pomset recogniser such that $\lang_\mathcal{R} = \lang_A$.
\end{lemma}

If $A$ is finite, then so is $\mathcal{R}$, since each of the elements of $M$ is a relation on $Q$, and there are finitely many relations on a finite set.

In general, the PA obtained from a pomset recogniser may admit runs where the same fork transition is nested repeatedly.
Recognisable pomset languages of \emph{bounded width} may be recognised by a pomset recogniser that is \emph{depth-nilpotent}~\cite{lodaya-weil-2000}, which can be converted into a \emph{fork-acyclic} PA by way of an sr-expression~\cite{lodaya-weil-2000,kappe-brunet-luttik-silva-zanasi-2018b}.
However, this detour via sr-expressions is not necessary: one can adapt \Cref{lemma:recognisable-to-saturated-regular} to produce a fork-acyclic PA, when given a depth-nilpotent pomset recogniser.
The details are discussed in
\ifarxiv%
\Cref{appendix:strongly-recognisable-to-strongly-regular}%
\else%
the full version~\cite{arxiv-version}%
\fi.

We conclude this section by remarking that the minimal pomset recogniser for a bounded-width language is necessarily depth-nilpotent~\cite{lodaya-weil-2000}; since our algorithm produces a minimal pomset recogniser, this means that we can also produce a fork-acyclic PA after learning a bounded-width recognisable pomset language.

\section{Discussion}\label{sec:discussion}

To learn DFAs, there are several alternatives to the observation table data structure that reduce the space complexity of the algorithm.
Most notable is the \emph{classification tree}~\cite{kearns1994}, which distinguishes individual pairs of words (which for us would be pomsets) at every node rather than filling an entire row for each of them.
The TTT algorithm~\cite{isberner2014} further builds on this and achieves optimal space complexity.
Given that we developed the first learning algorithm for pomset languages, we opted for the simplicity of the observation table---optimisations such as those analogous to the aforementioned work are left to future research.

We would like to extend our algorithm to learn recognisers based on arbitrary algebraic theories.
One challenge is to ensure that the equations of the theory hold for hypotheses, by generalising our definition of associativity (\Cref{def:assoctable}).

Our algorithm can also be specialised to learn languages recognised by commutative monoids.
These languages of \emph{multisets} can alternatively be represented as semi-linear sets~\cite{parikh-1966} or described using Presburger arithmetic~\cite{ginsburg-spanier-1964}.
While not all languages described this way are recognisable (for instance, the set of multisets over $\Sigma = \{ \ltr{a}, \ltr{b} \}$ with as many $\ltr{a}$'s as $\ltr{b}$'s~\cite{lodaya-weil-2000}), it would be interesting to be able to learn at least the fragment representable by commutative monoids, and apply that to one of the domains where semi-linear sets are used.

Our algorithm is limited to learning languages of series-parallel pomsets; there exist pomsets which are not series-parallel, each of which must contain an ``N-shape''~\cite{gischer-1988,grabowski-1981,valdes-tarjan-lawler-1982}.
Since N-shapes appear in pomsets that describe message passing between threads, we would like to be able to learn such languages as well.
We do not see an obvious way to extend our algorithm to include these pomsets, but perhaps recent techniques from~\cite{fahrenberg-johansen-struth-thapa-2020} can provide a solution.

Every hypothesis of our algorithm can be converted to a pomset automaton.
The final pomset recogniser for a bounded-width language is minimal, and hence depth-nilpotent~\cite{lodaya-weil-2000}, which means that it can be converted to a fork-acyclic PA\@.
In future work, we would like to guarantee that the same holds for intermediate hypotheses when learning a bounded-width language.

Running two threads in parallel may be implemented by running some initial section of those threads in parallel, followed by running the remainder of those threads in parallel.
This interleaving is represented by the \emph{exchange law}~\cite{gischer-1988,grabowski-1981}.
One can specialise pomset recognisers to include this interleaving to obtain recognisers of pomset languages closed under subsumption~\cite{lodaya-weil-2000}, i.e., such that if a pomset $u$ is recognised, then so are all of the ``more sequential'' versions of $u$.
We would like to adapt our algorithm to learn these types of recognisers, and exploit the extra structure provided by the exchange law to optimise further.

We have shown that recognisable pomset languages correspond to saturated regular pomset languages (\Cref{lemma:recognisable-to-saturated-regular,lemma:saturated-regular-to-recognisable}).
One question that remains is whether there is an algorithm that can learn all or at least a larger class of regular pomset languages.
Given that pomset automata can accept context-free languages (\Cref{figure:problematic-pa}), we wonder if a suitable notion of context-free grammars for pomset languages could be identified.
Clark~\cite{clark2010} showed that there exists a subclass of context-free languages that can be learned via an adaptation of {\LStar}.
Arguably, this adaptation learns recognisers with a monoidal structure and reverses this structure to obtain a grammar.
An extension of this work to pomset languages might lead to a learning algorithm that learns more PAs.

\bibliographystyle{splncs04}
\bibliography{pomset-learning}

\vfill

{\small\medskip\noindent{\bf Open Access} This chapter is licensed under the terms of the Creative Commons\break Attribution 4.0 International License (\url{http://creativecommons.org/licenses/by/4.0/}), which permits use, sharing, adaptation, distribution and reproduction in any medium or format, as long as you give appropriate credit to the original author(s) and the source, provide a link to the Creative Commons license and indicate if changes were made.}

{\small \spaceskip .28em plus .1em minus .1em The images or other third party material in this chapter are included in the chapter's Creative Commons license, unless indicated otherwise in a credit line to the material.~If material is not included in the chapter's Creative Commons license and your intended\break use is not permitted by statutory regulation or exceeds the permitted use, you will need to obtain permission directly from the copyright holder.}

\medskip\noindent\includegraphics{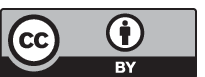}

\ifarxiv%

\appendix%

\section{Omitted Proofs about the Algorithm}\label{appendix:proofs-algo}

\begin{oneshot}{\Cref{lem:fullhyp}}
	$\angl{\SP, \PC}$ is a closed, consistent, and associative observation table.
\end{oneshot}
\begin{proof}
	Closedness holds trivially, and associativity follows from consistency.
	It remains to show that consistency holds.
	For $\asymbol \in \{{\cdot}, {\parallel}\}$ we have for all $s_1, s_2, t \in \SP$ such that $\row(s_1) = \row(s_2)$ and $e \in \PC$ that
	\begin{align*}
		\row(s_1 \asymbol t)(e) = 1 &
			\iff \cinsert{e}{s_1 \asymbol t} \in \lang \\
		&
			\iff \cinsert{\cinsert{e}{\cempty \asymbol t}}{s_1} \in \lang \\
		&
			\iff \row(s_1)(\cinsert{e}{\cempty \asymbol t}) = 1 \\
		&
			\iff \row(s_2)(\cinsert{e}{\cempty \asymbol t}) = 1 \\
		&
			\iff \cinsert{\cinsert{e}{\cempty \asymbol t}}{s_2} \in \lang \\
		&
			\iff \cinsert{e}{s_2 \asymbol t} \in \lang \\
		&
			\iff \row(s_2 \asymbol t)(e) = 1
	\intertext{%
        and symmetrically we can show that
    }
		\row(t \asymbol s_1)(e) = 1
            &\iff \row(t \asymbol s_2)(e) = 1. \tag*{\qed}
	\end{align*}
\end{proof}

\begin{oneshot}{\Cref{lem:reachability} (Hypothesis reachability)}
	Given a closed, consistent, and associative observation table $\angl{S, E}$, the hypothesis it induces is reachable.
	In particular, $\free{i_H}(s) = \row(s)$ for any $s \in S$.
\end{oneshot}
\begin{proof}
	By induction on the pomsets in $S$, using subpomset-closedness.
	We have
	\begin{mathpar}
    \free{i_H}(\pempty) = \unit_H = \row(\pempty)
	\and
    \forall \ltr{a} \in \Sigma.\ \free{i_H}(\ltr{a}) = i_H(\ltr{a}) = \row(\ltr{a}).
	\end{mathpar}
	Now assume for the inductive case that $s_1, s_2 \in S$ are such that $\free{i_H}(s_1) = \row(s_1)$ and $\free{i_H}(s_2) = \row(s_2)$.
	Then
	\[
		\free{i_H}(s_1 \cdot s_2) = \free{i_H}(s_1) \odot_H \free{i_H}(s_2) = \row(s_1) \odot_H \row(s_2) = \row(s_1 \cdot s_2)
	\]
	and similarly $\free{i_H}(s_1 \parallel s_2) = \row(s_1 \parallel s_2)$.
    \qed%
\end{proof}

\begin{oneshot}{\Cref{lem:nonce}}
	The hypothesis of any closed, consistent, and associative observation table $\angl{S, E}$ is $S$-compatible.
\end{oneshot}
\begin{proof}
	For all $s \in S$, we have
	\begin{align*}
		s \in \lang_\hyp &
			\iff \free{i_H}(s) \in F_H \\
		&
			\iff \row(s) \in F_H
			\tag{\Cref{lem:reachability}} \\
		&
			\iff \row(s)(\cempty) = 1 \\
		&
			\iff s \in \lang.
			\tag*{\qed}
	\end{align*}
\end{proof}

\begin{oneshot}{\Cref{lem:sound}}
	For all pomset recognisers $\angl{M, \odot, \obar, \unit, i, F}$ and $u, v \in \SP$ such that $\free{i}(u) = \free{i}(v)$ we have for any $c \in \PC$ that $\free{i}(\cinsert{c}{u}) = \free{i}(\cinsert{c}{v})$.
\end{oneshot}
\begin{proof}
	Proof by induction on the elements of $\PC$.
	For any $\ltr{a} \in \Sigma$ trivially have $\free{i}(\ltr{a}) = \free{i}(\ltr{a})$, and similarly $\free{i}(\pempty) = \free{i}(\pempty)$.
	Furthermore,
	\[
		\free{i}(\cinsert{\cempty}{u}) = \free{i}(u) = \free{i}(v) = \free{i}(\cinsert{\cempty}{v}).
	\]
	Now let $c_1, c_2 \in \PC$ satisfy $\free{i}(\cinsert{c_1}{u}) = \free{i}(\cinsert{c_1}{v})$ and $\free{i}(\cinsert{c_2}{u}) = \free{i}(\cinsert{c_2}{v})$.
	Then
	\begin{align*}
		\free{i}(\cinsert{(c_1 \cdot c_2)}{u}) &
			= \free{i}(\cinsert{c_1}{u} \cdot \cinsert{c_2}{u}) \\
		&
			= \free{i}(\cinsert{c_1}{u}) \odot \free{i}(\cinsert{c_2}{u}) \\
		&
			= \free{i}(\cinsert{c_1}{v}) \odot \free{i}(\cinsert{c_2}{v}) \\
		&
			= \free{i}(\cinsert{c_1}{v} \cdot \cinsert{c_2}{v}) \\
		&
			= \free{i}(\cinsert{(c_1 \cdot c_2)}{v})
	\end{align*}
	and similarly $\free{i}(\cinsert{(c_1 \parallel c_2)}{u}) = \free{i}(\cinsert{(c_1 \parallel c_2)}{v})$.
    \qed%
\end{proof}

\begin{oneshot}{\Cref{prop:canmin}}
	The canonical pomset recogniser is minimal.
\end{oneshot}
\begin{proof}
	From \Cref{lem:nonce} we know that the pomset recogniser accepts $\lang$, so by \Cref{prop:minimality} it is minimal.
\end{proof}

\section{Omitted proofs about saturated pomset automata}\label{appendix:proofs-pas}

\begin{oneshot}{\Cref{lemma:recognisable-to-saturated-regular}}%
Let $\mathcal{R} = \angl{M, \odot, \obar, \unit, i, F}$ be a pomset recogniser.
We construct the pomset automaton $A = \angl{M, F, \{ \unit \}, \delta, \gamma}$ (note: the set $F$ consists of initial states) where $\delta \colon M \times \Sigma \to 2^M$ and $\gamma \colon M \times \M(M) \to 2^M$ are given by
\begin{mathpar}
\delta(q, \ltr{a}) = \{ q' : f(\ltr{a}) \odot q' = q \}
\and
\gamma(q, \phi) = \{ q' : (r \obar r') \odot q' = q,\ \phi = \mset{r, r'} \}
\end{mathpar}
Then $A$ is saturated, and $\lang_A = \lang_\mathcal{R}$.
\end{oneshot}

\begin{proof}
Our proof relies on the following property of $A$:
\begin{claim}
$q \arun{u}_A q'$ if and only if $\free{i}(u) \odot q' = q$.
\end{claim}
\begin{proof}
For the direction from left to right, we proceed by induction on $\runrel_A$.
In the base, there are two cases.
On the one hand, suppose $q \arun{u}_A q'$ is a trivial run, i.e., $u = \pempty$ and $q = q'$.
We then calculate $\free{i}(u) \odot q' = \unit \odot q' = q' = q$.
On the other hand, suppose $q \arun{u}_A q'$ is a $\delta$-run, which is to say that $u = \ltr{a}$ for some $\ltr{a} \in \Sigma$ with $q' \in \delta(q, \ltr{a})$.
It then follows that, $\free{i}(u) \odot q' = i(\ltr{a}) \odot q' = q$.

For the inductive step, there are again two cases to consider.
\begin{itemize}
    \item
    Suppose that $q \arun{u}_A q'$ is a composite run, i.e., that $u = v \cdot w$ and $q'' \in M$ such that $q \arun{v}_A q''$ and $q'' \arun{w}_A q'$.
    By induction, we have $\free{i}(v) \odot q'' = q$ and $\free{i}(w) \odot q' = q''$.
    In total, we find that
    \[
        \free{i}(u) \odot q'
            = (\free{i}(v) \odot \free{i}(w)) \odot q'
            = \free{i}(v) \odot (\free{i}(w) \odot q')
            = \free{i}(v) \odot q''
            = q
    \]

    \item
    Suppose $q \arun{u} q'$ because $u = v \parallel w$ and there exist $r, r' \in M$ such that $q' \in \gamma(q, \mset{r, r'})$ as well as $r \arun{v}_A \unit$ and $r' \arun{w}_A \unit$.
    Then by induction we know that $\free{i}(v) = \free{i}(v) \odot \unit = r$ and $\free{i}(w) = \free{i}(w) \odot \unit = r'$.
    Furthermore, since $q' \in \gamma(q, \mset{r, r'})$ we know that $(r \obar r') \odot q' = q$.
    In total, we find that
    \[
        \free{i}(u) \odot q'
            = (\free{i}(v) \obar \free{i}(w)) \odot q'
            = (r \obar r') \odot q'
            = q
    \]
\end{itemize}

\noindent
For the other direction, we proceed by induction on the structure of $u$.
In the base, there are two cases.
On the one hand, if $u = \pempty$, then $q \arun{\pempty}_A q = q'$ immediately.
On the other hand, if $u = \ltr{a}$ for some $\ltr{a} \in \Sigma$, then $i(\ltr{a}) \odot q' = \free{i}(u) \odot q' = q$, and hence $q' \in \delta(q, \ltr{a})$.
Therefore, $q \arun{\ltr{a}}_A q'$.

For the inductive step, there are two cases to consider.
\begin{itemize}
    \item
    Suppose that $u = v \cdot w$ for $v, w \neq \pempty$.
    Let $q'' = \free{i}(w) \odot q'$.
    We then find by induction that $q'' \arun{w}_A q'$.
    Furthermore, note that $\free{i}(v) \odot q'' = \free{i}(v) \odot (\free{i}(w) \odot q') = (\free{i}(v) \odot \free{i}(w)) \odot q' = \free{i}(u) \odot q' = q$.
    Hence, we find by induction that $q \arun{v}_A q''$.
    Putting this together, it follows that $q \arun{u}_A q'$.

    \item
    If $u = v \parallel w$ for $v, w \neq \pempty$, then choose $r = \free{i}(v)$ and $r' = \free{i}(w)$, and note that $\free{i}(v) \odot \unit = r$ and $\free{i}(w) \odot \unit = r'$.
    Hence, we find by induction that $r \arun{v}_A \unit$ and $r' \arun{w}_A \unit$.
    Since also $(r \obar r') \odot q' = (\free{i}(v) \obar \free{i}(w)) \odot q' = \free{i}(u) \odot q' = q$, it follows that $q' \in \gamma(q, \mset{r, r'})$.
    In total, $q \arun{u}_A q'$.
    \qed%
\end{itemize}
\end{proof}

\noindent
Continuing the proof of \Cref{lemma:recognisable-to-saturated-regular}, to see that $A$ is saturated, let $u, v \in \SP$ with $u, v \neq 1$.
First, suppose $q \arun{u \cdot v}_A q'$.
By the above property, $\free{i}(u \cdot v) \odot q' = q$.
If we now choose $q'' = \free{i}(v) \odot q' \in M$, we find that $\free{i}(u) \odot q'' = \free{i}(u) \odot \free{i}(v) \odot q' = \free{i}(u \cdot v) \odot q'$.
By the same property we find that $q \arun{u}_A q''$ and $q'' \arun{v}_A q'$.

Next, suppose that $q \arun{u \parallel v}_A q'$.
By the property above, we have $(\free{i}(u) \obar \free{i}(v)) \odot q' = q$.
If we choose $r = \free{i}(u)$ and $s = \free{i}(v)$ as well as $r' = s' = \unit$, then $\free{i}(u) \odot r' = r$ and $\free{i}(v) \odot s' = s$, and thus $r \arun{u}_A r'$ and $s \arun{v}_A s'$ again by the same property.
Furthermore, since $(r \obar s) \odot q' = q$, we have $q' \in \gamma(q, \mset{r, s})$.

Note that $A$ is finite.
To see that $A$ accepts $\lang_\mathcal{R}$, note that $u \in \lang_\mathcal{R}$ precisely when $\free{i}(u) \in F$, which holds if and only if there exists a $q \in F$ with $\free{i}(u) = q$, which, by the above, is equivalent to $q \arun{u} \unit$ for some $q \in F$, i.e., $u \in \lang_A$.
\qed%
\end{proof}

\begin{oneshot}{\Cref{lemma:saturated-regular-to-recognisable}}
Let $A = \angl{Q, I, F, \delta, \gamma}$ be a saturated pomset automaton.
We can construct a pomset recogniser $\mathcal{R} = \angl{M, \odot, \obar, \unit, i, F'}$, where
\begin{mathpar}
M = \{ {\arun{u}_A} : u \in \SP \}
\and
{\arun{u}_A} \odot {\arun{v}_A} = {\arun{u \cdot v}_A}
\and
{\arun{u}_A} \obar {\arun{v}_A} = {\arun{u \parallel v}_A}
\and
i(\ltr{a}) = {\arun{\ltr{a}}_A}
\and
F' = \{ {\arun{u}_A} \in M : \exists q \in I, q' \in F.\ q \arun{u}_A q' \}
\end{mathpar}
Now $\odot$ and $\obar$ are well-defined, and $\mathcal{R}$ is a pomset recogniser such that $\lang_\mathcal{R} = \lang_A$.
\end{oneshot}
\begin{proof}
Without loss of generality, we can assume that for $u \in \SP$, we have ${\arun{u}_A} = {\arun{\pempty}_A}$ if and only if $u = \pempty$: the implication from right to left is obvious, and the converse can be guaranteed (while preserving saturation of $A$ as well as its language) by adding a non-initial and non-accepting state without transitions.

Let $u, u', v, v' \in \SP$ such that ${\arun{u}_A} = {\arun{u'}_A}$ and ${\arun{v}_A} = {\arun{v'}_A}$.
To prove that the operations are well-defined, we should show ${\arun{u \cdot v}_A} = {\arun{u' \cdot v'}_A}$ and ${\arun{u \parallel v}_A} = {\arun{u' \parallel v'}_A}$.
For the former equality, we consider two cases.
\begin{itemize}
    \item
    If $u = \pempty$, then ${\arun{u'}_A} = {\arun{u}_A} = {\arun{\pempty}_A}$, and thus $u' = \pempty$.
    In that case
    \[
        {\arun{u \cdot v}_{A'}}
            = {\arun{v}_{A'}}
            = {\arun{v'}_{A'}}
            = {\arun{u' \cdot v'}_{A'}}
    \]
    A similar derivation applies if any of the other pomsets are empty.

    \item
    If $u, u', v, v' \neq \pempty$, suppose $q \arun{u \cdot v}_A q'$.
    Because $A$ is saturated, we find a $q'' \in Q$ such that $q \arun{u}_A q''$ and $q'' \arun{v}_A q'$.
    Since ${\arun{u}_A} = {\arun{u'}_A}$ and ${\arun{v}_A} = {\arun{v'}_A}$, it follows that $q \arun{u'}_A q''$ and $q'' \arun{v'}_A q'$, and hence $q \arun{u' \cdot v'}_A q'$.
    This shows that $\arun{u \cdot v}_A$ is contained in $\arun{u' \cdot v'}_A$; the converse can be shown symmetrically.
\end{itemize}

\noindent
Next, we show that ${\arun{u \parallel v}_A} = {\arun{u' \parallel v'}_A}$; again, we have two cases to consider.
\begin{itemize}
    \item
    If $u = \pempty$, then ${\arun{u'}_A} = {\arun{u}_A} = {\arun{\pempty}_A}$, and thus $u' = \pempty$ by the above.
    Thus,
    \[
        {\arun{u \parallel v}_{A'}}
            = {\arun{v}_{A'}}
            = {\arun{v'}_{A'}}
            = {\arun{u' \parallel v'}_{A'}}
    \]
    A similar derivation applies if any of the other pomsets are empty.

    \item
    If $u, u', v, v' \neq \pempty$, suppose that $q \arun{u \parallel v}_A q'$.
    In that case, since $A$ is saturated, we find $r, s \in Q$ and $r', s' \in F$ such that $r \arun{u}_A r'$ and $s \arun{v}_A s'$, and $q' \in \gamma(q, \mset{r, s})$.
    Since ${\arun{u}_A} = {\arun{u'}_A}$ and ${\arun{v}_A} = {\arun{v'}_A}$, it then follows that $r \arun{u'}_A r'$ and $s \arun{v'} s'$, and thus $q \arun{u' \parallel v'} q'$.
    This shows that $\arun{u \parallel v}_A$ is contained in $\arun{u' \parallel v'}$; the converse inclusion follows by a similar argument.
\end{itemize}

\noindent
The bimonoid laws can now be proved straightforwardly; for instance, $\obar$ is associative because for $u, v, w \in \SP$ we have
\[
    {\arun{u}_A} \obar ({\arun{v}_A} \obar {\arun{w}_A})
        = \arun{u \parallel (v \parallel w)}_A
        = \arun{(u \parallel v) \parallel w}_A
        = ({\arun{u}_A} \obar {\arun{v}_A}) \obar {\arun{w}_A}
\]
Associativity of $\odot$ and the fact that $\arun{\pempty}_A$ is a unit can be shown similarly.
This makes $\mathcal{R}$ a proper pomset recogniser.

It should be clear that for $u \in \SP$ we have $\free{i}(u) = {\arun{u}_A}$.
From this, it follows that $u \in \lang_A$ if and only if there exist $q \in I$ and $q' \in F$ such that $q \arun{u}_A q'$, which holds precisely when $\free{i}(u) = {\arun{u}_A} \in F'$, and thus $u \in \lang_\mathcal{R}$.
\qed%
\end{proof}

\section{Translation to fork-acyclic pomset automata}%
\label{appendix:strongly-recognisable-to-strongly-regular}

Bimonoids recognisers can, in general, encode pomset languages of \emph{unbounded width}, i.e., without an upper bound on the antichains in the pomsets, and even pomset languages of \emph{unbounded depth}, with a complex mutual nesting structure between sequential and parallel composition.
By \Cref{lemma:recognisable-to-saturated-regular}, even these can be accepted by a pomset automaton.
Typical programs, however, have a limited number of parallel threads, and tend not to launch threads recursively.
To further restrict PAs so as to exclude this kind of behaviour, we can use the notion of fork-acyclicity~\cite{lodaya-weil-2000}; for pomset automata, this comes down to the following~\cite{kappe-brunet-luttik-silva-zanasi-2018b}:

\begin{definition}[Fork-acyclicity]
Let $A = \angl{Q, I, F, \delta, \gamma}$ be a PA\@.
We define the \emph{support relation} of $A$, denoted $\preceq_A$, as the smallest preorder on $Q$ satisfying
\begin{mathpar}
\inferrule{%
    \ltr{a} \in \Sigma \\
    q' \in \delta(q, \ltr{a})
}{
    q' \preceq_A q
}
\and
\inferrule{%
    \phi \in \M(Q) \\
    q' \in \gamma(q, \phi)
}{%
    q' \preceq_A q
}
\and
\inferrule{%
    \gamma(q, \mset{r, \ldots})
}{%
    r \preceq_A q
}
\end{mathpar}
We say that $A$ is \emph{fork-acyclic} if for all $q, r \in Q$ and $\phi \in \M(Q)$ with $r \in \phi$ and $\gamma(q, \phi) \neq \emptyset$ it holds that $q \not\preceq_A r$.
\end{definition}

\begin{example}
Recall the PAs drawn in \Cref{figure:example-pas}.
If $A$ is the PA in \Cref{figure:simple-pa}, then $A$ is fork-acyclic; after all, the only fork is given by $q_2 \in \gamma(q_1, \mset{q_3, q_4})$, and we have that $q_3, q_4 \prec_A q_1$---in other words, the runs starting at $q_3$ or $q_4$ do not depend on $q_1$.
On the other hand, if $A$ is the PA in \Cref{figure:problematic-pa}, then $A$ is not fork-acyclic, since $\gamma(q_1, \mset{q_3, q_4}) \neq \emptyset$ but $q_1 \preceq_A q_3$, since $q_1 \in \delta(q_3, \ltr{a})$.
\end{example}

Intuitively, if $q$ and $q'$ are states such that $q' \preceq_A q$, then the pomsets that can be read starting in $q$ somehow depend on those that can be read from $q'$.
A pomset automaton is fork-acyclic if it can only fork into states whose language does not depend on the point of origin for the fork.
For example, the PA in \Cref{figure:simple-pa} is fork-acyclic, because $q_3$ and $q_4$ cannot reach (or fork into) $q_1$, while the PA in \Cref{figure:problematic-pa} is not, because $\gamma(q_1, \mset{q_3, q_4}) \neq \emptyset$ while $q_1 \preceq_A q_3$.

The corresponding restriction on pomset recognisers is \emph{depth-nilpotency}~\cite{lodaya-weil-2000}, which we adapt for our purposes as follows.

\begin{definition}[Depth-nilpotency]
Let $\mathcal{R} = \angl{M, \odot, \obar, \unit, i, F}$ be a pomset recogniser.
We define $\prec_\mathcal{R}$ as the smallest transitive relation on $M$ satisfying the following rule for all $s, u, v, w, t, x, y \in M$:
\[
    \inferrule{%
        s = u \odot (v \obar (w \odot t \odot x)) \odot y \\
        v \obar (w \odot t \odot x) \neq w \odot t \odot x
    }{%
        s \prec_\mathcal{R} t
    }
\]
We say that $\mathcal{R}$ is \emph{depth-nilpotent} if the following hold:
\begin{enumerate}[label={(\roman*)}]
    \item
    there exists an $N \in \naturals$ such that every $\prec_\mathcal{R}$-chain is of length at most $N$, that is, for all $q_1, \dots, q_n \in M$ with $q_1 \prec_\mathcal{R} \cdots \prec_\mathcal{R} q_n$, it holds that $n \leq N$.

    \item
    there exists a $\zero \in M \setminus F$ such that $s \parallel \zero = \zero$ for all $s \in M$.

    \item
    if $s, t \in M$ and $s \obar t = t$, then either $s = \unit$ or $t = \zero$.

    \item
    for all $u \in \SP$, we have $\free{i}(u) = \unit$ if and only if $u = \pempty$.
\end{enumerate}
When $\mathcal{R}$ is depth-nilpotent, we write $D_\mathcal{R}(s)$ for the length of the longest $\prec_\mathcal{R}$-chain starting at $s \in M$, i.e., the maximal $n$ such that there exist $s = s_1, \dots, s_n \in M$ with $s_1 \prec_\mathcal{R} \cdots \prec_\mathcal{R} s_n$.
\end{definition}

\begin{example}
The pomset recogniser $\mathcal{R}$ defined in \Cref{ex:loop} is depth-nilpotent; indeed, the maximal $\prec_\mathcal{R}$-chain is given by $q_\bot \prec_\mathcal{R} q_1 \prec_\mathcal{R} q_\ltr{a}, q_\ltr{b} \prec_\mathcal{R} \unit$.

On the other hand, if $\mathcal{R}$ is as in \Cref{ex:nested}, then $\mathcal{R}$ is \emph{not} depth-nilpotent, because $q_\ltr{b} = q_\ltr{a} \odot (q_\ltr{b} \obar q_\ltr{b})$, and hence $q_\ltr{b} \prec_\mathcal{R} q_\ltr{b}$.
\end{example}

One can combine existing results to show that a language recognised by a depth-nilpotent pomset recogniser can also be accepted by a fork-acyclic PA:\@ by~\cite[Theorem~3.9]{lodaya-weil-2000}, every such language is series-rational, and by~\cite[Theorem~7.16]{kappe-brunet-luttik-silva-zanasi-2018b}, every series-rational language is recognised by a fork-acyclic PA\@.
However, this detour is not necessary: we can convert a depth-nilpotent bimonoid to a fork-acyclic automaton directly, by adapting the construction from \Cref{lemma:recognisable-to-saturated-regular}.

\begin{lemma}%
\label{lemma:strongly-recognisable-to-strongly-regular}
Let $\mathcal{R} = \angl{M, \odot, \obar, \unit, i, F}$ be a depth-nilpotent pomset recogniser.
Let $A = \angl{M, F, \{ \unit \}, \delta, \gamma}$ be the PA constructed from $\mathcal{R}$ just like in \Cref{lemma:recognisable-to-saturated-regular}, except that we define the parallel transition function $\gamma \colon: M \times \M(M) \to 2^M$ by
\[
    \gamma(q, \phi) = \{ q' : (r \obar s) \odot q' = q,\, \phi = \mset{r, s},\, D_\mathcal{R}(r) < D_\mathcal{R}(q),\, D_\mathcal{R}(s) < D_\mathcal{R}(q) \}
\]
Now $A$ is fork-acyclic, and $\lang_A = \lang_\mathcal{R}$.
\end{lemma}

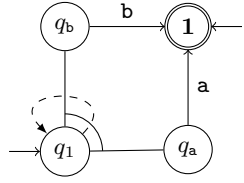
\begin{figure}[t]
    \centering
    \begin{tikzpicture}[initial text={}]
        \node[state] (qb) {$q_\ltr{b}$};
        \node[state,initial right,accepting,right=of qb] (unit) {$\unit$};
        \node[state,initial,below=of qb] (q1) {$q_1$};
        \node[state,below=of unit] (qa) {$q_\ltr{a}$};

        \draw[->] (qb) edge node[above] {$\ltr{b}$} (unit);
        \draw[->] (qa) edge node[right] {$\ltr{a}$} (unit);
        \draw (q1) edge (qb);
        \draw (q1) edge (qa);
        \draw pic[draw,angle radius=.5cm] {angle=qa--q1--qb};
        \path (q1) edge[dashed,-latex,loop,looseness=5] (q1);
    \end{tikzpicture}\vspace{2.5mm} % chktex 31
    \caption{%
        Part of the PA obtained from the pomset recogniser from \Cref{ex:loop}, using the construction from \Cref{lemma:strongly-recognisable-to-strongly-regular}.
        Once more, the state $q_\bot$ and transitions into $\unit$ are not pictured.
    }%
    \label{figure:strongly-recognisable-to-strongly-regular}
\end{figure}

\begin{example}
Let $\angl{M, \odot, \obar, \unit, i, F}$ be the pomset recogniser from \Cref{ex:loop}.
The pomset automaton that arises from the construction above is partially depicted in \Cref{figure:strongly-recognisable-to-strongly-regular}.
Here, we have that $q_1$ and $\unit$ are initial states, since they both appear in $F$.
Furthermore, $\unit \in \delta(q_\ltr{b})$ because $i(\ltr{b}) \odot \unit = q_\ltr{b}$, and $q_1 \in \gamma(q_1, \mset{q_\ltr{a}, q_\ltr{b}})$ because $(q_\ltr{a} \obar q_\ltr{b}) \odot q_1 = q_1 \odot q_1 = q_1$.
However, $\unit \not\in \gamma(q_\ltr{b}, \mset{q_\ltr{b}, \unit})$ despite $(q_\ltr{b} \obar \unit) \odot \unit = q_\ltr{b}$, because $D_\mathcal{R}(q_\ltr{b})$ is not strictly larger than $D_\mathcal{R}(q_\ltr{b})$.
\end{example}

\begin{proof}[of \Cref{lemma:strongly-recognisable-to-strongly-regular}]
Our proof rests on the following two properties of $A$.

\begin{claim}
If $\free{i}(u) \odot q' = q$ and $q \neq \mathbf{0}$, then $q \arun{u}_A q'$.
\end{claim}
\begin{proof}
We proceed by induction on $u$.
In the base, we have two cases.
On the one hand, if $u = \pempty$, then $\free{i}(u) = \unit$, and so the claim holds immediately.
On the other hand, if $u = \ltr{a}$ for some $\ltr{a} \in \Sigma$, then $q' \in \delta(q, \ltr{a})$, whence $q \arun{u}_A q'$ again.

For the inductive step, there are two cases.
\begin{itemize}
    \item
    If $u = v \cdot w$ with $v, w \neq \pempty$, then choose $q'' = \free{i}(w) \odot q'$.
    Now $q'' \neq \mathbf{0}$, otherwise $q = \free{i}(v) \odot q'' = \mathbf{0}$.
    By induction, we then find that $q'' \arun{w}_A q'$, as well as $q \arun{v}_A q''$.
    In total, we have $q \arun{u}_A q'$.

    \item
    If $u = v \parallel w$ with $v, w \neq \pempty$, then choose $r = \free{i}(v)$ and $s = \free{i}(w)$.
    Now $r, s \neq \mathbf{0}$, otherwise $q = \free{i}(u) \odot q = (\free{i}(v) \obar \free{i}(w)) \odot q' = \mathbf{0}$.
    Furthermore, $r, s \neq \unit$, otherwise $v = \pempty$ or $w = \pempty$.
    Together, this means that $r \obar s \neq r$ and $r \obar s \neq s$, by depth-nilpotency of $M$.
    Since $\free{i}(v) \odot \unit = r$ and $\free{i}(w) \odot \unit = s$, we find by induction that $r \arun{v}_A \unit$ and $s \arun{w}_A \unit$.
    Since $q = (r \obar s) \odot q'$, it follows that $q \prec_\mathcal{R} r, s$, and hence $D_\mathcal{R}(r) < D_\mathcal{R}(q)$ as well as $D_\mathcal{R}(s) < D_\mathcal{R}(q)$.
    We then know that $q' \in \gamma(q, \mset{r, s})$, and thus $q \arun{u}_A q'$.
    \qed%
\end{itemize}
\end{proof}

\begin{claim}
If $q \preceq_A q'$, then $D_\mathcal{R}(q') \leq D_\mathcal{R}(q)$.
\end{claim}
\begin{proof}
It suffices to validate the claim for the pairs that generate $\preceq_A$.
\begin{itemize}
    \item
    If $q' \preceq_A q$ because $q' \in \delta(q, \ltr{a})$ for some $\ltr{a} \in \Sigma$, then $\free{i}(\ltr{a}) \odot q' = q$.
    Now suppose that $q' \prec_\mathcal{R} q''$ for some $q'' \in M$; we then obtain $u, v, w, x, y \in M$ such that $q' = u \odot (v \obar (w \odot q'' \odot x)) \odot y$, and $v \obar (w \odot q'' \odot x) \neq w \odot q'' \odot x$.
    Clearly, $q = \free{i}(\ltr{a}) \cdot q' = \free{i}(\ltr{a}) \odot u \odot (v \obar (w \odot q'' \odot x)) \odot y$, and hence $q \prec_\mathcal{R} q''$.
    It follows that every chain starting at $q'$ can be turned into one starting at $q$, hence $D_\mathcal{R}(q') \leq D_\mathcal{R}(q)$.

    \item
    If $q' \preceq_A q$ because $q' \in \gamma(q, \mset{r, s})$ for some $r, s \in M$, then a similar argument to the previous case shows that $D_\mathcal{R}(q') \leq D_\mathcal{R}(q)$.

    \item
    If $q' \preceq_A q$ because $\gamma(q, \mset{q', r}) \neq \emptyset$ for some $r \in M$, then $D_\mathcal{R}(q') < D_\mathcal{R}(q)$ by definition of $\gamma$.
    \qed%
\end{itemize}
\end{proof}

We return to the proof of \Cref{lemma:strongly-recognisable-to-strongly-regular}.
If $u \in \lang_A$, then there exists a $q \in F$ with $q \arun{u}_A \unit$.
By the same argument as in \Cref{lemma:recognisable-to-saturated-regular}, $\free{i}(u) \odot \unit = q$---after all, the transitions of the new automaton are also transitions of the old automaton.
Hence $\free{i}(u) \in F$, meaning $u \in \lang_\mathcal{R}$.
Conversely, if $u \in \lang_\mathcal{R}$, then $\free{i}(u) \neq \zero$ by depth-nilpotency, and by the first claim above we have $q \arun{u}_A \unit$, meaning $u \in \lang_A$.

For fork-acyclicity, let $q, r, s \in M$ be such that $\gamma(q, \mset{r, s}) \neq \emptyset$.
In that case, $D_\mathcal{R}(r) < D_\mathcal{R}(q)$ by definition; thus $q \preceq_A r$ cannot hold, because that would imply $D_\mathcal{R}(r) \leq D_\mathcal{R}(r)$ by the second claim above.
\qed%
\end{proof}

The converse also exists in the literature: the language of a fork-acyclic PA can also be recognised by a depth-nilpotent pomset recogniser, because any fork-acyclic automaton can be converted to an equivalent sr-expression~\cite[Theorem~8.4]{kappe-brunet-luttik-silva-zanasi-2018b}, and every sr-expression can in turn be converted to a depth-nilpotent pomset recogniser~\cite[Theorem~3.9]{lodaya-weil-2000}.
A direct construction is also possible, if we use the techniques from~\cite{kappe-brunet-luttik-silva-zanasi-2018b} to show that every fork-acyclic PA can be converted to an equivalent fork-acyclic PA that is also saturated, and reuse the construction from \Cref{lemma:saturated-regular-to-recognisable}; we omit this proof for the sake of brevity.

\fi%

\end{document}